\documentclass[runningheads, envcountsame]{llncs}
\pdfoutput=1

\usepackage{configuration}
\usepackage{pdfpages}

\lstdefinelanguage{JSON}{
    autogobble=true,
    basicstyle=\normalfont\footnotesize\ttfamily,
    numbers=left,
    numberstyle=\scriptsize,
    stepnumber=1,
    showstringspaces=false,
    breaklines=true,
    frame=lines,
    literate=
     *{0}{{{\color{numb}0}}}{1}
      {1}{{{\color{numb}1}}}{1}
      {2}{{{\color{numb}2}}}{1}
      {3}{{{\color{numb}3}}}{1}
      {4}{{{\color{numb}4}}}{1}
      {5}{{{\color{numb}5}}}{1}
      {6}{{{\color{numb}6}}}{1}
      {7}{{{\color{numb}7}}}{1}
      {8}{{{\color{numb}8}}}{1}
      {9}{{{\color{numb}9}}}{1}
      {:}{{{\color{punct}{:}}}}{1}
      {,}{{{\color{punct}{,}}}}{1}
      {\{}{{{\color{delim}{\{}}}}{1}
      {\}}{{{\color{delim}{\}}}}}{1}
      {[}{{{\color{delim}{[}}}}{1} % chktex 9
      {]}{{{\color{delim}{]}}}}{1}, % chktex 9
    morestring=[b]", % chktex 18
    morestring=[d]',
}

\newif\ifshortversion\shortversionfalse
\newif\ifanonymous\anonymousfalse
    
\def\orcidID#1{\smash{\href{http://orcid.org/#1}{\protect\raisebox{-1.25pt}{\protect\includegraphics{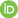}}}}}

\title{Validating Streaming JSON Documents with Learned VPAs%
    \thanks{This work was supported by the Belgian FWO \enquote{SAILor} project (G030020N). Ga\"etan Staquet is a research fellow (Aspirant) of the Belgian F.R.S.-FNRS.}%
}
\author{%
    V\'{e}ronique Bruy\`{e}re\inst{1}\orcidID{0000-0002-9680-9140}%
\and Guillermo A. P\'{e}rez\inst{2}\orcidID{0000-0002-1200-4952}%
\and Ga\"{e}tan Staquet\inst{1,2}\orcidID{0000-0001-5795-3265}%
}
\authorrunning{V. Bruy\`ere, G. A. P\'erez, and G. Staquet}

\institute{%
    University of Mons (UMONS), Mons, Belgium\\
	\email{\{veronique.bruyere,gaetan.staquet\}@umons.ac.be}
    \and%
    University of Antwerp (UAntwerp) -- Flanders Make, Antwerp, Belgium % chktex 8
	\email{guillermo.perez@uantwerpen.be}%
}

\begin{document}
    \maketitle
    
    \begin{abstract}
        We present a new streaming algorithm to validate JSON documents against a set of constraints given as a JSON schema. Among the possible values a JSON document can hold, objects are unordered collections of key-value pairs while arrays are ordered collections of values. We prove that there always exists a visibly pushdown automaton (VPA) that accepts the same set of JSON documents as a JSON schema.
        Leveraging this result, our approach relies on learning a VPA for the provided schema.
        As the learned VPA assumes a fixed order on the key-value pairs of the objects, we abstract its transitions in a special kind of graph, and propose an efficient streaming algorithm using the VPA and its graph to decide whether a JSON document is valid for the schema.
        We evaluate the implementation of our algorithm on a number of random JSON documents, and compare it to the classical validation algorithm.
    \end{abstract}

\setcounter{tocdepth}{3}

\section{Introduction}
%======================

     \emph{JavaScript Object Notation} (JSON) has overtaken XML as the de facto standard data-exchange
     format, in particular for web applications. JSON documents are easier to read for programmers and end users since they only have arrays and objects as structured types. Moreover, in contrast to XML\@, they do not include named open and end tags for all values,
     but open and end tags (braces actually) for arrays and objects only.
     \emph{JSON schema}~\cite{JSONSchemaSite} is a simple schema language that allows users to
     impose constraints on the structure of JSON documents.

     In this work, we are interested in the \emph{validation} of
     \emph{streaming} JSON documents against JSON schemas. Several previous
     results have been obtained about the formalization of XML schemas and
     the use of formal methods to validate XML documents (see, e.g.,~\cite{DBLP:journals/mst/BonevaCS15,DBLP:conf/www/KumarMV07,mnns17,ns18,schwentick12,DBLP:conf/pods/SegoufinV02}).
     Recently, a standard to formalize JSON schemas has been proposed and
     (hand-coded) validation tools for such schemas can be found
     online~\cite{JSONSchemaSite}. Pezoa et al,
     in~\cite{DBLP:conf/www/PezoaRSUV16}, observe that the standard of JSON
     documents is still evolving and that the formal semantics of JSON
     schemas is also still changing. Furthermore, validation tools seem to
     make different assumptions about both documents and schemas. The
     authors of~\cite{DBLP:conf/www/PezoaRSUV16} carry out an initial
     formalization of JSON schemas into formal grammars from which they are
     able to construct a \emph{batch} validation tool from a given JSON
     schema.
        
     In this paper, we rely on the formalization work of~\cite{DBLP:conf/www/PezoaRSUV16} and propose a \emph{streaming} algorithm
     for validating JSON documents against JSON schemas. To our knowledge, this is the first JSON validation algorithm that is streaming.
     For XML, works that study streaming document
     validation base such algorithms on the construction of some
     automaton (see, e.g.,~\cite{DBLP:conf/pods/SegoufinV02}, for
     XML).
     \ifanonymous
     The authors of~\cite{DBLP:conf/tacas/BruyerePS22} first experimented with one-counter automata for this purpose.
     \else
     In~\cite{DBLP:conf/tacas/BruyerePS22}, we first experimented with
     one-counter automata for this purpose.
     \fi
     We submit that
     \emph{visibly-pushdown automata} (VPAs) are a better fit for this task
     --- this is in line with~\cite{DBLP:conf/www/KumarMV07}, where the same
     was proposed for streaming XML documents. In contrast to one-counter
     automata,\footnote{By nesting objects and arrays, we obtain a set of JSON documents encoding $\{a^n b^m c^m d^n \mid n,m \in \mathbb{N}\}$, a context-free language that requires two counters.}
     we show that VPAs are expressive enough to capture the language of JSON
     documents satisfying any JSON schema.
        
     More importantly, we explain that \emph{active learning} \textit{\`a la}
     Angluin~\cite{DBLP:journals/iandc/Angluin87} is a good alternative to the
     automatic construction of such a VPA from the formal semantics of a given
     JSON schema. This is possible in the presence of labeled examples or a
     computer program that can answer membership and (approximate) equivalence
     queries about a set of JSON documents. This learning approach has two
     advantages. First, we derive from the learned VPA a streaming validator
     for JSON documents. Second, by automatically learning an automaton
     representation, we circumvent the need to write a schema and subsequently
     validate that it represents the desired set of JSON documents. Indeed, it
     is well known that one of the highest bars that users have to clear to
     make use of formal methods is the effort required to write a formal
     specification, in this case, a JSON schema.
      
     \paragraph{Contributions.} We present a VPA active learning framework to
     achieve what was mentioned above --- though we fix an order on the keys
     appearing in objects. The latter assumption helps our algorithm learn faster.
     Secondly, we show how to bootstrap the learning algorithm by leveraging
     existing validation and document-generation tools to implement
     approximate equivalence checks.
     Thirdly, we describe how to validate streaming documents using our fixed-order learned automata --- that is, our algorithm accepts other permutations of keys, not just the one
     encoded into the VPA\@.
     Finally, we present an empirical evaluation of our learning and validation algorithms, implemented on top of \LearnLib~\cite{DBLP:conf/tacas/MertenSHM11}.
     %We highlight that this current work differs from~\cite{DBLP:conf/www/KumarMV07} thanks to these contributions.
        
     All contributions, while complementary, are valuable in their
     own right. First, our learning algorithm for VPAs is a novel gray-box
     extension of TTT~\cite{DBLP:phd/dnb/Isberner15} that leverages side
     information about the language of all JSON documents. Second, our
     validation algorithm that uses a fixed-order VPA is novel and can be
     applied regardless of whether the automaton is learned or constructed
     from a schema.
     For the validation algorithm, we developed the concept of \emph{key graph}, which allows us to efficiently realize the validation no matter the key-value order in the document, and might be of independent interest for other JSON-analysis applications using VPAs.
     Finally, we implemented our own batch validator to
     facilitate approximating equivalence queries as required by our learning
     algorithm. Both the new validator and the equivalence
     oracle are efficient, open-source, and easy to modify. We strongly
     believe the latter can be re-used in similar projects aiming to learn
     automata representations of sets of JSON documents.

\ifshortversion
    A long version of this work is on arXiv: \url{https://arxiv.org/abs/2211.08891}.
\else
    \paragraph{Structure of the paper.} 
    %The paper is structured as follows.
        In \Cref{sec:vpa}, we recall the concept of VPA and the structure we impose for our learning context.
        Namely, we work on \emph{1-single entry VPAs}~\cite{DBLP:conf/icalp/AlurKMV05,DBLP:phd/dnb/Isberner15} where all call transitions lead to the initial state, and that admit a unique minimal automaton.
        We also recall how to learn such an automaton using TTT~\cite{DBLP:phd/dnb/Isberner15}.
        In \Cref{sec:JSONformat}, we introduce JSON documents and JSON schemas, and define a formal grammar encoding a schema under the abstractions we consider.
        Moreover, we prove that for any JSON schema, we can construct a VPA that accepts the same set of JSON documents, which implies that our learning approach is feasible.
        Then, in \Cref{sec:validation}, we present our new streaming validation algorithm based on learned VPAs.
        This algorithm relies on a specific graph that abstracts the transitions of the VPA with respect to the objects that appear in the accepted JSON documents.
        We define this graph and prove several useful properties before providing the validation algorithm, for which we prove time and space complexity results and its correctness.
        Finally, in \Cref{sec:implementation} we discuss the implementation of our algorithms and show experimental results comparing our validation approach with the ones proposed in~\cite{JSONSchemaSite}.
\fi

\section{Visibly Pushdown Languages}\label{sec:vpa}
%===================================

First, we recall the definition VPAs~\cite{DBLP:conf/stoc/AlurM04} %After some preliminaries, we provide the definition of this class of automata 
and state some of their %important
properties. We 
%finally 
also recall how they can be actively learned following Angluin's approach~\cite{DBLP:journals/iandc/Angluin87}.

\subsection{Preliminaries}\label{subsec:prelim}
%========================

An \emph{alphabet} $\alphabet$ is a finite set whose elements are called \emph{symbols}. A \emph{word} $w$ over $\alphabet$ is a finite sequence of symbols from $\alphabet$, with the \emph{empty word} denoted by $\emptyword$.
The length of $w$ is denoted $\lengthOf{w}$; the set of all words, $\alphabet^*$.
Given two words $v,w \in \alphabet^*$, $v$ is a \emph{prefix} (resp.\ \emph{suffix}) of $w$ if there exists $u \in \alphabet^*$ such that $w = vu$ (resp. $w = uv$), and $v$ is a \emph{factor} of $w$ if there exist $u,u' \in \alphabet^*$ such that $w = uvu'$.
Given $L \subseteq \alphabet^*$, called a \emph{language}, we denote by $\prefixes{L}$ (resp. $\suffixes{L}$) the set of prefixes (resp.\ suffixes) of words of $L$.

Given a set $Q$, we write $\identity{Q}$ for the \emph{identity relation} $\{ (q,q) \mid q \in Q\}$ on $Q$.
%being the \emph{identity relation} on $Q$.
  
\subsection{Visibly Pushdown Automata}\label{subsec:VPAs} 
%====================================
  
VPA~\cite{DBLP:conf/stoc/AlurM04} are particular pushdown automata that we recall in this section. %operating on a partitioned alphabet where only call symbols can push, return symbols can pop, and internal symbols can do transitions without considering the stack. 
The \emph{pushdown alphabet}, denoted $\pushdownAlphabet = (\callAlphabet, \returnAlphabet, \internalAlphabet)$, is partitioned into  pairwise disjoint alphabets $\callAlphabet, \returnAlphabet, \internalAlphabet$ such that $\callAlphabet$ (resp. $\returnAlphabet$, $\internalAlphabet$) is the set of \emph{call} symbols (resp.\ \emph{return} symbols, \emph{internal} symbols). In this paper, we work with the particular alphabet of return symbols
%equal to  
$\returnAlphabet = \{\bar a \mid a \in \callAlphabet\}$. For any such $\pushdownAlphabet$, we denote by $\alphabet$ the alphabet $\callAlphabet \cup \returnAlphabet \cup \internalAlphabet$. Given a pushdown alphabet $\pushdownAlphabet$, the set $\wellMatched$ of \emph{well-matched} words over $\pushdownAlphabet$ is defined:
%by induction:
%as follows:
\begin{itemize}
    \item $\emptyword \in \wellMatched$,
    \item $\forall a \in \internalAlphabet, a \in \wellMatched$,
    \item if $w,w' \in \wellMatched$, then $ww' \in \wellMatched$,
    \item if $a \in \callAlphabet, w \in \wellMatched$, then $aw\bar a \in \wellMatched$.
\end{itemize}
Also, the \emph{call/return balance} function $\balance : \alphabet^* \to \Z$ is defined as %$\balance(\emptyword) = 0$ and $\balance(wa) = \balance(w) + x$ with $x$ being $1$, $-1$, or $0$ if $a$ is in $\callAlphabet$, $\returnAlphabet$, or $\internalAlphabet$ respectively.
\[
    \balance(ua) = \balance(u) + \begin{cases}
        \phantom{-}1    & \text{if \(a \in \callAlphabet\),}\\
        -1              & \text{if \(a \in \returnAlphabet\),}\\
        \phantom{-}0    & \text{if \(a \in \internalAlphabet\).}
    \end{cases}
\]
In particular, for all $w \in \wellMatched$, we have $\balance(u) \geq 0$ for each prefix $u$ of $w$ and $\balance(u) \leq 0$ for each suffix $u$ of $w$. Finally, the \emph{\depth} $\depthSymbol(w)$ of a well-matched word $w$ is equal to $\max \{\balance(u) \mid u \in \prefixes{\{w\}}\}$, that is, the maximum number of unmatched call symbols among the prefixes of $w$.

\begin{definition}\label{def:VPA}
A \emph{visibly pushdown automaton} (VPA) % $\automaton$ 
over a pushdown alphabet $\pushdownAlphabet$ is a tuple %$\automaton = 
$(\states, \pushdownAlphabet, \stackAlphabet, \transitionFunction, \initialStates, \finalStates)$ where \(\states\) is a finite non-empty set of \emph{states}, \(\initialStates \subseteq \states\) is a set of \emph{initial} states, \(\finalStates \subseteq \states\) is a set of \emph{final} states,  \(\stackAlphabet\) is a \emph{stack alphabet}, and  $\transitionFunction$ is a finite set of \emph{transitions} of the form \(\transitionFunction = \callFunction \cup \returnFunction \cup \internalFunction\) where
\begin{itemize}
    \item $\callFunction \subseteq \states \times \callAlphabet \times \states \times \stackAlphabet$ is the set of \emph{call} transitions,
    \item $\returnFunction \subseteq \states \times \returnAlphabet \times \stackAlphabet \times \states$ is the set of \emph{return} transitions,  %where $\emptystack$ is a special bottom-of-stack symbol,
    \item $\internalFunction \subseteq \states \times \internalAlphabet \times \states$ is the set of \emph{internal} transitions.
\end{itemize}
The \emph{size} $\automatonSize$ of a VPA $\automaton$ is its number of states. Its number of transitions is denoted by $|\transitionFunction|$. 
\end{definition}

Let us describe the \emph{transition system} $\transitionSystem$ of a VPA $\automaton$ whose vertices are configurations. % defined as follows. 
A \emph{configuration} is a pair $\left\langle q,\stackWord \right\rangle$ where $q\in Q$ is a state and $\stack\in\Gamma^*$ a stack content. A configuration is \emph{initial} (resp.\ \emph{final}) if $q\in \initialStates$ (resp. $q\in \finalStates$) and $\stackWord = \emptyword$. For $a \in \alphabet$, we write $\left\langle q,\stackWord \right\rangle \xrightarrow{a} \left\langle q',\stackWord' \right\rangle$ in $\transitionSystem$ if there is:
\begin{itemize}
    \item a call transition $(q,a,q',\gamma) \in \callFunction$ verifying $\stackWord' = \gamma\stackWord$,\footnote{The stack symbol $\gamma$ is pushed on the left of $\stackWord$.}
    \item a return transition $(q,a,\gamma,q') \in \returnFunction$ verifying $\stackWord = \gamma\stackWord'$,
    \item an internal transition $(q,a,q') \in \internalFunction$ such that $\stackWord' = \stackWord$.
\end{itemize}
The transition relation of $\transitionSystem$ is extended to words in the usual way. We say that $\automaton$ \emph{accepts} a word $w \in \alphabet^*$ if there exists a path in $\transitionSystem$ from an initial configuration to a final configuration that is labeled by $w$. The \emph{language of} \(\automaton\), denoted by $\languageOf{\automaton}$, is defined as the set of all words accepted by $\automaton$:
             \[
                \languageOf{\automaton} = \left\{w \in \alphabet^* \:\middle|\: \exists q \in \initialStates, \exists q' \in \finalStates, \left\langle q, \emptyword \right\rangle \xrightarrow{w} \left\langle q', \emptyword \right\rangle\right\}.
            \]
Any language accepted by some VPA is a \emph{visibly pushdown language} (VPL). Notice that such a language is composed of well-matched words only.\footnote{The original definition of VPA~\cite{DBLP:conf/stoc/AlurM04} allows acceptance of ill-matched words.} Given a VPA $\automaton$ over $\pushdownAlphabet$, the \emph{reachability relation} $\accRelation$ of $\automaton$ is: 
            \[
                \accRelation = \left\{(q,q') \in \states^2 \:\middle|\: \exists w \in \wellMatched, \left\langle q, \emptyword \right\rangle \xrightarrow{w} \left\langle q', \emptyword \right\rangle \right\}.
            \]

Finally, we say that $p \in \states$ is a \emph{\binState{}} if there exists no path in $\transitionSystem$ of the form $\left\langle q, \emptyword \right\rangle \xrightarrow{w} \left\langle p, \stackWord \right\rangle \xrightarrow{w'} \left\langle q', \emptyword \right\rangle$ with $q \in \initialStates$ and $q' \in \finalStates$. If a VPA $\automaton$ has \binStates{}, those states can be removed from $\states$ as well as the transitions containing \binStates{} without modifying the accepted language. 

\begin{example}
Let \(\pushdownAlphabet = (\{a\}, \{\bar a\}, \{b\})\) be a pushdown alphabet and \(\stackAlphabet = \{\gamma\}\) be a stack alphabet.
An example of a VPA \(\automaton\) with three states is given in \cref{fig:vpa}.
Since \(a\) is a call symbol, the transition going from \(q_0\) to \(q_1\) indicates that the symbol \(\gamma\) must be pushed on the stack.
On the other hand, the return transition reading \(\bar a\) from \(q_1\) to \(q_2\) can only be triggered if the top of the stack is \(\gamma\).

The word \(aba\bar a^2\) is accepted by \(\automaton\).
Indeed, we have the following sequence of configurations:
\(
    (q_0, \emptyword) \transition^a (q_1, \gamma) \transition^b (q_0, \gamma) \transition^a (q_1, \gamma^2) \transition^{\bar a} (q_2, \gamma) \transition^{\bar a} (q_2, \emptyword),
\)
i.e., a $w$-labeled path
in $\transitionSystem$ from an initial configuration to a final configuration.
\end{example}

\begin{figure}[t]
    \centering
    \begin{tikzpicture}[
    automaton
]
    \node [state, initial]                  (q0)    {\(q_0\)};
    \node [state, right=of q0]              (q1)    {\(q_1\)};
    \node [state, accepting, right=of q1]   (q2)    {\(q_2\)};
    
    \path
        (q0)    edge [bend left]    node {\(a, \gamma\)} (q1)
        (q1)    edge    node {\(b\)}    (q0)
                edge                node {\(\bar a, \gamma\)} (q2)
        (q2)    edge [loop right]   node {\(\bar a, \gamma\)} (q2)
    ;
\end{tikzpicture}
    \caption{A VPA accepting the language \(\{a{(ba)}^n \bar a^{n+1} \mid n \in \N\}\).}%
    \label{fig:vpa}
\end{figure}
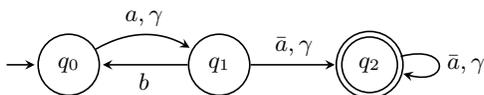

\subsection{Minimal Deterministic VPAs}\label{subsec:MinimalVPAs}
%=========================================

Given a VPA $\automaton = (\states, \pushdownAlphabet, \stackAlphabet, \transitionFunction, \initialStates, \finalStates)$, we say that it is \emph{deterministic} (det-VPA) if $|\initialStates| = 1$ and $\automaton$ does not have two distinct transitions with the same left-hand side. By \emph{left-hand side}, we mean $(q,a)$ for a call transition $(q,a,q',\gamma) \in \callFunction$ or an internal transition $(q,a,q') \in \internalFunction$, and $(q,a,\gamma)$ for a return transition $(q,a,\gamma,q') \in \returnFunction$. %VPAs can always be determinized.

\begin{theorem}[\cite{DBLP:conf/stoc/AlurM04,DBLP:journals/corr/abs-0911-3275}]\label{thm:detVPA}
    For any VPA $\automaton$ over $\pushdownAlphabet$, one can construct a det-VPA $\automaton[B]$ over $\pushdownAlphabet$ such that $\languageOf{\automaton} = \languageOf{\automaton[B]}$. Moreover, the size of $\automaton[B]$ is in $\complexity(2^{\automatonSize^2})$ and the size of its stack alphabet is in $\complexity(|\callAlphabet|\cdot 2^{\automatonSize^2})$.
\end{theorem}

\begin{proof}
Let us briefly recall this construction. Let $\automaton = (\states, \pushdownAlphabet, \stackAlphabet, \transitionFunction, \initialStates, \finalStates)$. The states of $\automaton[B]$ are subsets $\setOfStates$ of the reachability relation $\accRelation$ of $\automaton$ and the stack symbols of $\automaton[B]$ are of the form $(\setOfStates,a)$ with $\setOfStates \subseteq \accRelation$ and $a \in \callAlphabet$. Let $w = u_1a_1u_2a_2 \ldots u_n a_n u_{n+1}$ be such that $n \geq 0$ and $u_i \in \wellMatched, a_i \in \callAlphabet$ for all $i$. That is, we decompose $w$ in terms of its unmatched call symbols. Let $\setOfStates_i$ be equal to $\{(p,q) \mid  \langle p, \emptyword \rangle \xrightarrow{u_i} \langle q, \emptyword \rangle\}$ for all $i$. Then after reading $w$, the det-VPA $\automaton[B]$ has its current state equal to $\setOfStates_{n+1}$ and its stack containing $(\setOfStates_n,a_n) \ldots (\setOfStates_2,a_2)(\setOfStates_1,a_1)$. Assume we are reading the symbol $a$ after $w$, then $\automaton[B]$ performs the following transition from $\setOfStates_{n+1}$:
\begin{itemize}
    \item if $a \in \callAlphabet$, then push $(\setOfStates_{n+1},a)$ on the stack and go to the state $\setOfStates = \identity{\states}$ (a new unmatched call symbol is read),
    \item if $a \in \internalAlphabet$, then go to the state $\setOfStates = \{(p,q) \mid \exists (p,p') \in \setOfStates_{n+1}, (p',a,q) \in \internalFunction \}$ ($u_{n+1}$ is extended to the well-matched word $u_{n+1}a$),
    \item if $a \in \returnAlphabet$, then pop $(\setOfStates_n,a_n)$ from the stack if $\bar a_n =  a$, and go to the state $\setOfStates = \{(p,q) \mid \exists (p,p') \in \setOfStates_n, (p',a_n,r',\gamma) \in \callFunction, (r',r) \in \setOfStates_{n+1}, (r,a,\gamma,q) \in \returnFunction \}$ (the call symbol $a_n$ is matched with the return symbol $a = \bar a_n$, leading to the well-matched word $u_n a_n u_{n+1}a$).
\end{itemize}
Finally the initial state of $\automaton[B]$ is $\identity{\initialStates}$ and its final states are sets $\setOfStates$ containing some $(p,q)$ with $p \in \initialStates$ and $q \in \finalStates$.
\qed\end{proof}

The family of VPLs is closed under several natural operations.
%Boolean operations, and concatenation and Kleene-$*$ operations.

\begin{theorem}[\cite{DBLP:conf/stoc/AlurM04}]\label{thm:closeness}
Let $L_1$ and $L_2$ be two VPLs over $\pushdownAlphabet$. Then, $L_1 \cup L_2$, $L_1 \cap L_2$, $\wellMatched \setminus L_1$, $L_1  L_2$, and $L^*$ are VPLs over $\pushdownAlphabet$. %f f is a renaming of  ̃Σ to  ̃Σ′, then f (L1) is a visibly pushdown language with respect to  ̃Σ′.
\end{theorem}

Though a VPL $L$ in general does not have a unique minimal det-VPA $\automaton$ accepting $L$, imposing the following subclass
%of 1-module single entry VPAs
leads to a unique minimal acceptor.

\begin{definition}[\cite{DBLP:conf/icalp/AlurKMV05,DBLP:phd/dnb/Isberner15}]\label{def:1SEVPA}
    A \emph{1-module single entry VPA}\footnote{The definitions
    of 1-SEVPA %given in
    in~\cite{DBLP:conf/icalp/AlurKMV05} and~\cite{DBLP:phd/dnb/Isberner15} %are slightly different. In this paper, we follow the definition provided in \cite{DBLP:phd/dnb/Isberner15}.}
    differ slightly. We follow the one in~\cite{DBLP:phd/dnb/Isberner15}.}
    (1-SEVPA) is a det-VPA $\automaton = (\states, \pushdownAlphabet, \stackAlphabet, \transitionFunction, \initialStates = \{q_0\}, \finalStates)$ such that
    %: 
    %\begin{itemize}
    %    \item  
    its stack alphabet $\stackAlphabet$ is equal to $\states \times \callAlphabet$, and
    %    \item  
    all its call transitions $(q,a,q',\gamma) \in \callFunction$ are such that $q' = q_0$ and $\gamma = (q,a)$.
    %\end{itemize}
\end{definition}

\begin{theorem}[\cite{DBLP:conf/icalp/AlurKMV05}]\label{thm:minimal1SEVPA}
    For any VPL $L$, there exists a unique minimal (with regards to the number of states) 1-SEVPA accepting \(L\), up to a renaming of the states.\footnote{This 1-SEVPA may be exponentially bigger than the size of a VPA accepting \(L\).}
\end{theorem}

Let us remark two facts about minimal 1-SEVPAs. First, given a minimal 1-SEVPA $\automaton$, there may exist a smaller VPA accepting the same language
%as $\automaton$ 
(that is therefore not a 1-SEVPA). Second, the transition relation of $\automaton$ is a total function, meaning that $\automaton$ may have a \binState{} (that is unique, in case of existence).

\subsection{Learning VPAs}\label{sec:definitions:learning}
%========================

Let us recall the concept of \emph{learning} a deterministic finite automaton (DFA), as introduced in~\cite{DBLP:journals/iandc/Angluin87}. %Let $\Sigma$ be an alphabet and 
Let \(L\) be a regular language over an alphabet $\Sigma$. The task of the \emph{learner} is to construct a DFA \(\hypothesis\) such that \(\languageOf{\hypothesis} = L\) by interacting with the \emph{teacher}.
The two possible types of interactions are \emph{membership queries} (does \(w \in \Sigma^*\) belong to \(L\)?), and \emph{equivalence queries} (does the DFA \(\hypothesis\) accept \(L\)?). For the latter type, if the answer is negative, the teacher also provides a counterexample, i.e., a word \(w\) such that \(w \in L \iff w \notin \languageOf{\hypothesis}\). The so-called \emph{$L^*$ algorithm} of~\cite{DBLP:journals/iandc/Angluin87} learns at least one representative per equivalence class of the Myhill-Nerode congruence of \(L\)~\cite{hu00} from which the minimal DFA \(\cal D\) accepting \(L\) is constructed. This learning process terminates and it uses a polynomial number of membership and equivalence queries in the size of \(\cal D\), and in the length of the longest counterexample returned by the teacher~\cite{DBLP:journals/iandc/Angluin87}.

In~\cite{DBLP:phd/dnb/Isberner15}, an efficient learning algorithm for VPLs is given by extending Angluin's learning algorithm. The Myhill-Nerode congruence for regular languages is extended to VPLs as follows. Given a pushdown alphabet \(\pushdownAlphabet\) and a VPL $L$ over \(\pushdownAlphabet\), we consider the set $\contextPairs$ of \emph{context pairs}\footnote{Notice that a non-empty word in $(\wellMatched \cdot \callAlphabet)^*$ is an element of $\prefixes{\wellMatched}$ ending with a call symbol.}
            \[
                \contextPairs = \left\{(u, v) \in {(\wellMatched \cdot \callAlphabet)}^* \times ~\suffixes{\wellMatched} \:\middle|\: \balance(u) = -\balance(v)\right\},
            \]
and we define the equivalence relation \(\VPLRelation \subseteq \wellMatched \times \wellMatched\)~\cite{DBLP:conf/icalp/AlurKMV05,DBLP:phd/dnb/Isberner15} such that %for all \(w, w' \in \wellMatched\) we have 
$w \VPLRelation w'$ if and only if $\forall (u, v) \in \contextPairs, uwv \in L \iff u w' v \in L$.

The minimal 1-SEVPA accepting $L$ as described in \Cref{thm:minimal1SEVPA} is constructed from %the equivalence relation 
$\VPLRelation$ such that its states are the equivalence classes of $\VPLRelation$.

\begin{theorem}[\cite{DBLP:phd/dnb/Isberner15}]\label{thm:learningVPA}
Let $L$ be a VPL over $\pushdownAlphabet$ and $n$ be the index of $\VPLRelation$.
%. Suppose that the equivalence relation $\VPLRelation$ has $n$ equivalence classes. 
Then, one can learn the minimal 1-SEVPA accepting $L$ with at most $n-1$ equivalence queries and a number of membership queries polynomial in $n$, $|\Sigma|$, and $\log \ell$, where
$\ell$ is the length of the longest counterexample returned by the teacher. 
\end{theorem}
The learning process designed in~\cite{DBLP:phd/dnb/Isberner15} extends to VPLs the \emph{\TTT{} algorithm} proposed in~\cite{DBLP:conf/rv/IsbernerHS14} for regular languages. 
\TTT{} improves the efficiency of the $L^*$ algorithm by eliminating redundancies in counterexamples provided by the teacher.   

\section{JSON Format}\label{sec:JSONformat}
%====================

In this section, we first describe JSON documents~\cite{DBLP:journals/rfc/rfc8259}. This format is currently the most popular one used for exchanging information on the web. We also describe JSON schemas that impose some constraints on the structure of JSON documents. These schemas can thus be seen as a formalism defining subsets of JSON documents. Then, for the purpose of this paper, we make some abstractions on JSON documents and schemas, and we show how to construct a VPA that accepts the set of JSON documents defined by a given JSON schema. We finally explain how to learn such a VPA from the schema.

\subsection{JSON Documents}
%==========================

We describe here the structure of JSON documents. Our presentation is inspired by~\cite{DBLP:conf/www/PezoaRSUV16} such that some details are skipped for readability (we refer to the official website~\cite{JSONDotOrg} for a full description).  

The JSON format defines six different types of \emph{JSON values}: 
\begin{itemize}
    \item \verb!true!, \verb!false! are JSON values.
    \item \verb!null! is a JSON value. 
    \item Any decimal number (positive, negative) is a JSON value, called a \emph{number}. In particular any number that is an integer is called an \emph{integer}. 
    \item Any finite sequence of Unicode characters starting and ending with \verb!"! is a JSON value, called a \emph{string value}.
    
    \item If $v_1, v_2, \ldots, v_n$ are JSON values and $k_1, k_2, \ldots, k_n$ are \emph{pairwise distinct} string values, then 
    \(
        \{ k_1 \!\!:\!\! v_1 , k_2 \!\!:\!\! v_2 , \ldots , k_n \!\!:\!\! v_n \}
    \)
    is a JSON value, called an \emph{object}. Each $k_i \!\!:\!\! v_i$ is called a \emph{key-value pair} such that $k_i$ is the \emph{key}. The collection of these pairs is \emph{unordered}.
    
    \item If $v_1, v_2, \ldots, v_n$ are JSON values, then 
    \(
        [v_1, v_2, \ldots, v_n]
    \)
    is a JSON value, called an \emph{array}. Each $v_i$ is an \emph{element} and the collection thereof is \emph{ordered}.
\end{itemize}
The first four types of JSON values are called \emph{\basicValues{}}. In this work, \emph{JSON documents} are supposed to be objects.\footnote{In~\cite{DBLP:journals/rfc/rfc8259}, a JSON document can be any JSON value and duplicated keys are allowed inside objects. In this paper, we follow what is commonly used in practice: JSON documents are objects, and keys are pairwise distinct inside objects.}

\begin{example}\label{ex:JSONdocument}
    An example of a JSON document is given in \Cref{fig:json_document}.
    We can see that this document is an object containing three keys: \verb!"title"!, whose associated value is a string value; \verb!"keywords"!, whose value is an array containing string values; and \verb!"conference"!, whose value is an object.
    This inner object contains two keys: \verb!"name"!, whose value is a string value; \verb!"year"!, whose value is an integer.

\begin{lstlisting}[language=JSON, caption=A JSON document., label=fig:json_document]
{
  "title": "Validating Streaming JSON Documents with Learned VPAs",
  "keywords": ["visibly pushdown automata", "JSON documents", "streaming validation"],
  "conference": {
    "name": "TACAS",
    "year": 2023
  }
}
\end{lstlisting}
\end{example}

It is possible to navigate through JSON documents. If $J$ is an object and $k$ is a key, then $J[k]$ is the value $v$ such that the key-value pair $k\!:\!v$ appears in $J$. If $J$ is an array and $n$ is a natural number, then $J[n]$ is the $(n+1)$-th element of $J$.\footnote{Arrays are indexed from 0.} More generally, values can be retrieved from JSON documents by using \emph{JSON pointers}:
%that are 
sequences of references as defined previously.
For instance, in \Cref{ex:JSONdocument}, \verb!J[keywords][1]!, where \texttt{J} is the root of the document, allows to retrieve the value \verb!"JSON documents"!.

\subsection{JSON Schemas}\label{sec:definitions:schema}
%=======================

A \emph{JSON schema} can impose some constraints on JSON documents by specifying any of the six types of JSON values that appear in those documents. We say that a JSON document \emph{satisfies} the schema if it verifies the constraints imposed by this schema. In this section, we give a simplified presentation of JSON schemas and refer to~\cite{JSONSchemaSite} for a complete description and to~\cite{DBLP:conf/www/PezoaRSUV16} for a formalization (i.e.\ a formal grammar with its syntax and semantics). 

A JSON schema is itself a JSON document that uses several keywords that help shape and restrict the set of JSON documents that this schema specifies:\footnote{A more formal description is given in the sequel.}
\begin{itemize}
    \item It can be imposed that a string value has a minimum/maximum length or satisfies a pattern expressed by a regular expression;
    \item That a number belongs to some interval or is a multiple of some number.
    \item Within object schemas, restrictions can be imposed on the key-value pairs of the objects. For example, the value associated with some key has itself to satisfy a certain schema, or some particular keys must be present in the object.
    A minimum/maximum number of pairs can also be imposed.
    \item Within array schemas, it can be imposed that all elements of the array satisfy a certain schema, or that the array has a minimum/maximum size.
    \item Schemas can be combined with Boolean operations, in the sense that a JSON document must satisfy the conjunction/disjunction of several JSON schemas, or it must not satisfy a certain JSON schema.
    \item A schema can be the enumeration of certain JSON values.
    \item A schema can be defined as one referred to by a JSON pointer. This allows a \emph{recursive} structure for the JSON documents satisfying a certain schema.
\end{itemize}

\begin{example}\label{ex:JSONschema}
    The schema from \Cref{fig:json_schema}
    describes the objects that can have three keys: \verb!"title"!, whose associated value must be a string value; \verb!"keywords"!, whose value must be an array containing string values; and \verb!"conference"!, whose value must be an object.
    Among these, \verb!"title"! and \verb!"conference"! are required.
    The JSON document of \Cref{ex:JSONdocument} satisfies this JSON schema.

\begin{lstlisting}[language=JSON, caption=A JSON schema., label=fig:json_schema]
{
  "type": "object",
  "required": ["title", "conference"],
  "properties": {
    "title": { "type": "string" },
    "keywords": {
      "type": "array",
      "items": { "type": "string" }
    },
    "conference": {
      "type": "object",
      "required": ["name", "year"],
      "properties": {
        "name": { "type": "string" },
        "year": { "type": "integer" }
      }
    }
  }
}
\end{lstlisting}
\end{example}

\subsection{Abstract JSON Documents and Schemas}\label{subsec:AbstractJSON}
%================================

For the purpose of this paper, we consider somewhat \emph{abstract} JSON values, documents, and schemas. We see JSON values as well-matched words over the pushdown alphabet $\JSONpushdownAlphabet$ that we describe hereafter. 
We do not consider the restrictions that can be imposed on string values and numbers, and we abstract all string values as \verb!s!, and all numbers as \verb!n! (as \verb!i! when they are integers). We denote by $\valAlphabet = \{\verb!true!,\verb!false!,\verb!null!,\verb!s!,\verb!n!,\verb!i!\}$ the alphabet composed of the \basicValues{}. We also do not consider enumerations. Concerning the key-value pairs appearing in objects, each key together with the symbol \enquote{$\doubleDot$} following the key is abstracted as an alphabet symbol~$k$. We assume knowledge of a \emph{finite} alphabet $\keyAlphabet$ of keys. The pushdown alphabet $\JSONpushdownAlphabet = (\callAlphabet,\returnAlphabet,\internalAlphabet)$ that we use is such that:
\begin{itemize}
    \item $\internalAlphabet = \keyAlphabet \cup \valAlphabet \cup \{\comma\}$, where $\comma$ is used in place of the comma,
    \item $\callAlphabet = \{ \laccol, \lcrochet \}$, where $\laccol$ (resp. $\lcrochet$) is used in place of \enquote{$\{$} (resp.\ \enquote{$[$}), % chktex 9
    \item $\returnAlphabet = \{ \raccol, \rcrochet \}$, with $\overline \laccol = \raccol$ and $\overline \lcrochet =  \rcrochet$.
\end{itemize}
We denote by $\JSONalphabet$ the set $\callAlphabet \cup \returnAlphabet \cup \internalAlphabet$.

\begin{example}\label{ex:AbstractDoc}
The JSON document given in \Cref{ex:JSONdocument} is abstracted as the word 
\[
    \laccol k_1 \verb!s! \comma k_2 \lcrochet \verb!s! \comma \verb!s! \comma \verb!s! \rcrochet \comma k_3 \laccol k_4 \verb!s! \comma k_5 \verb!i! \raccol \raccol \in \wellMatched[\JSONpushdownAlphabet]
\]
where $\keyAlphabet$ contains the keys $k_i, i \in \{1,\ldots,5\}$.
\end{example}

We now use the formalism of \emph{extended context-free grammars} (extended CFGs) to define JSON schemas with the abstractions mentioned previously. We recall that in an extended CFG, the right-hand sides of productions are regular expressions over the terminals and non-terminals. Here, we even use \emph{generalized} regular expressions such that intersections and negations are allowed in addition to union, concatenation and Kleene-$*$ operations. An extended CFG $\grammar$ defining a JSON schema uses the alphabet $\JSONalphabet$ of terminals, an alphabet $\nonTerminalAlphabet = \{S,S_1,S_2, \ldots\}$ of non-terminals (with one of them being the axiom), and a finite series of productions as described in the next definition.

\begin{definition}\label{def:JSONsyntax}
Productions used in abstract JSON schemas are of the form:\footnote{In an abuse of notation, the non-terminals of $\nonTerminalAlphabet$ are also called schemas.}\begin{description} 
    \item[Primitive schema] $\schema \Coloneqq v$ with $v \in \valAlphabet$. 
    \item[Object schema] $\schema \Coloneqq \laccol k_1 \schema_1 \comma k_2 \schema_2 \comma \ldots \comma k_n \schema_n \raccol$, with $n \geq 0$ and pairwise distinct keys $k_i \in \keyAlphabet$, $i \in \{1,\ldots, n\}$.
    \item[Array schema] either $\schema \Coloneqq \lcrochet \emptyword \vee \schema_1 (\comma \schema_1)^* \rcrochet$, or $\schema \Coloneqq \lcrochet \schema_1 \comma \ldots \comma \schema_1\rcrochet$ with $n\geq 0$ occurrences of $\schema_1$. 
    \item[Boolean operations] $\schema \Coloneqq \schema_1  \vee \ldots \vee \schema_n$, $\schema \Coloneqq \schema_1 \wedge \ldots \wedge \schema_n$, and $\schema \Coloneqq \neg \schema_1$.
\end{description} 
\end{definition}
This extended CFG $\grammar$ must be \emph{closed}, i.e., whenever it contains a production $\schema \Coloneqq \laccol k_1 \schema_1 \comma k_2 \schema_2 \comma \ldots \comma k_n \schema_n \raccol$, then it also contains \emph{all productions} $\schema \Coloneqq \laccol k_{i_1} \schema_{i_1} \comma k_{i_2} \schema_{i_2} \comma \ldots \comma k_{i_n} \schema_{i_n} \raccol$ where $(i_1,\ldots,i_n)$ is a permutation of $(1,\ldots,n)$. Indeed we recall that objects are unordered collections of key-value pairs. 

\begin{example}\label{ex:abstractJSONschema}
The JSON schema of \Cref{ex:JSONschema} can be defined as follows:
\begin{eqnarray}
\nonTerminal_0 &\Coloneqq& \laccol k_1 \nonTerminal_1 \comma k_2 \nonTerminal_2 \comma k_3 \nonTerminal_3 \raccol\\
\nonTerminal_0 &\Coloneqq& \laccol k_1 \nonTerminal_1 \comma k_3 \nonTerminal_3 \raccol \\
\nonTerminal_1 &\Coloneqq&  \verb!s!\\
\nonTerminal_2 &\Coloneqq&  \lcrochet \emptyword \vee \nonTerminal_1 (\comma \nonTerminal_1)^* \rcrochet\\
\nonTerminal_3 &\Coloneqq& \laccol k_4 \nonTerminal_1 \comma k_5 \nonTerminal_4 \raccol\\
\nonTerminal_4 &\Coloneqq&  \verb!i!
\end{eqnarray}
where we add to the first, second, and fifth productions all the related productions with key permutations. The axiom of this grammar is $S_0$.
\end{example}

Let us explain the \emph{semantics} of a closed extended CFG $\grammar$ defining a JSON schema. Let $S \in \nonTerminalAlphabet$ be a non-terminal and $J \in \wellMatched[\JSONpushdownAlphabet]$ be a JSON value. We say $J$ \emph{satisfies} $\schema$, or $J$ is \emph{valid} with regard to $\schema$, denoted by $J \models \schema$, if one of the following holds:
\begin{itemize}
    \item $\schema$ is a primitive schema $v$ and $J = v$.
    \item $\schema$ is an object schema $\laccol k_1 \schema_1 \comma k_2 \schema_2 \comma \ldots \comma k_n \schema_n \raccol$, $J$ is an object $\laccol k_1 v_1 \comma k_2 v_2 \comma \ldots \comma k_n v_n \raccol$ such that $v_i \models \schema_i$ for every $i \in \{1,\ldots,n\}$. 
    \item $\schema$ is an array schema $\lcrochet \emptyword \vee \schema_1 (\comma \schema_1)^* \rcrochet$, $J$ is an array, and for each element $v$ of $J$, we have $v \models \schema_1$. 
    \item $\schema$ is an array schema $\lcrochet \schema_1 \comma \ldots \comma \schema_1\rcrochet$ with $n\geq 0$ occurrences of $\schema_1$, $J$ is an array of size $n$, and for each element $v$ of $J$, we have $v \models \schema_1$. 
    \item $\schema$ is $\schema_1 \vee \schema_2 \vee \ldots \vee \schema_n$ and there exists $i \in \{1,\ldots,n\}$ such that $J \models \schema_i$.
    \item $\schema$ is $\schema_1 \wedge \schema_2 \wedge \ldots \wedge \schema_n$ and for all $i \in \{1,\ldots,n\}$ we have $J \models \schema_i$.
    \item $\schema$ is $\neg \schema_1$ and $J \not\models \schema_1$.
\end{itemize}

Let us make some comment about this semantics. If $J$ is an object satisfying $S$ with respect to the production $\schema \Coloneqq \laccol k_1 \schema_1 \comma k_2 \schema_2 \comma \ldots \comma k_n \schema_n \raccol$, as $\grammar$ is closed, then the document $J$ with any permutation of its key-value pairs also satisfies $S$ as $\grammar$ contains all the related productions with key permutations. Notice also that an array $J$ can be composed of any number of elements or of a fixed number of elements, all satisfying the same schema. Moreover, the empty object $\laccol \raccol$ and the empty array $\lcrochet \rcrochet$ are allowed.

We denote by $\languageOf{\grammar}$ the set of all JSON values $J$ satisfying the axiom $S$ of $\grammar$. Hence, when a JSON schema, given as a grammar $\grammar$, defines a set of JSON documents,  this set of documents is the language $\languageOf{\grammar}$.\footnote{In the sequel, we abuse notation and speak of JSON values, documents and schemas instead of abstractions.}

\begin{example}
Consider the closed extended CFG $\grammar$ of \Cref{ex:abstractJSONschema}.
The JSON document $J$ of \Cref{ex:AbstractDoc}, equal to
\[
    \laccol k_1 \verb!s! \comma k_2 \lcrochet \verb!s! \comma \verb!s! \comma \verb!s! \rcrochet \comma k_3 \laccol k_4 \verb!s! \comma k_5 \verb!i! \raccol \raccol
\]
satisfies the axiom $S_0$ of $\grammar$. This is also the case for the document $J'$ equal to 
\[
    \laccol k_3 \laccol k_5 \verb!i! \comma k_4 \verb!s!  \raccol \comma k_2 \lcrochet \verb!s! \comma \verb!s! \comma \verb!s! \rcrochet \comma k_1 \verb!s! \raccol
\]
Therefore, $J, J'$ are both valid and $J, J' \in \languageOf{\grammar}$.
\end{example}

The set of all JSON values can be defined by a particular grammar as given in the next lemma.
\begin{lemma}\label{lem:universal}
The set of all JSON values $J \in \wellMatched[\JSONpushdownAlphabet]$ is equal to $\languageOf{\grammarUniv}$ where $\grammarUniv$ is the following closed extended CFG, called \emph{universal}, with $\nonTerminalV$ being the axiom:
\begin{itemize}
    \item $\nonTerminalV \Coloneqq v$  for all  $v \in \valAlphabet$,
    \item $\nonTerminalV \Coloneqq \lcrochet \emptyword \vee \nonTerminalV (\comma \nonTerminalV)^* \rcrochet$,
    \item $\nonTerminalV \Coloneqq \laccol k_1 \nonTerminalV \comma k_2 \nonTerminalV \comma \ldots \comma k_n \nonTerminalV \raccol$ for all sequences $(k_1,\ldots,k_n)$, $n \geq 0$, of pairwise distinct keys of $\keyAlphabet$.
%    \item $\nonTerminalV \Coloneqq \nonTerminal_1 \vee \nonTerminal_2 \vee \nonTerminal_3$.
\end{itemize}
 \end{lemma}

\begin{remark}\label{rem:well-formed}
As mentioned in~\cite{DBLP:conf/www/PezoaRSUV16}, we suppose to work with \emph{well-formed} extended CFGs that avoid problematic situations like in the production $\nonTerminal \Coloneqq \neg \nonTerminal$ or the productions $\nonTerminal \Coloneqq \nonTerminal_1 \vee \nonTerminal_2$, $\nonTerminal_2 \Coloneqq \nonTerminal$.
That is, we avoid grammars with cyclic definitions inside the productions that involve Boolean operations (see~\cite{DBLP:conf/www/PezoaRSUV16} for more details).
\end{remark}

\subsection{VPAs for JSON Schemas}\label{subsec:VPAandJSON}
%================================

We now explain how VPAs can be used in place of JSON schemas. Let us first introduce some definitions.

\begin{definition}
Let $\grammar$ be a closed extended CFG defining a JSON schema. Let $\languageOf{\grammar}$ be the set of JSON documents over $\JSONpushdownAlphabet$ satisfying this schema. Given an order $<$ of $\keyAlphabet$, the language $\languageOrderedOf{\grammar}$ is the subset of $\languageOf{\grammar}$ composed of the JSON documents whose key order inside objects respects the order $<$ of $\keyAlphabet$.
\end{definition}

In \Cref{sec:implementation}, we discuss our implementation of a new algorithm validating JSON documents against a given JSON schema presented in \Cref{sec:validation}. A key order is naturally induced by the data structures used in this implementation.

\begin{example}
Consider again the JSON schema of \Cref{ex:abstractJSONschema}. Let us order the alphabet $\keyAlphabet$ such that $k_1 < k_2 < k_3 < k_4 < k_5$. The JSON document $J$ of \Cref{ex:AbstractDoc} belongs to $\languageOrderedOf{\grammar}$ but none of the documents equal to $J$ up to any permutation of its key-value pairs belong to $\languageOrderedOf{\grammar}$ (permutations of the three pairs with keys $k_1, k_2, k_3$ and of the two pairs with keys $k_4, k_5$).
\end{example}

We have the next property, whose proof is given further below.

\begin{theorem}\label{thm:VPAforJSON}
Let $\grammar$ be a closed extended CFG defining a JSON schema. Then, there exists a VPA $\automaton$ such that $\languageOf{\automaton} = \languageOf{\grammar}$. Moreover for all orders $<$ of $\keyAlphabet$, there exists a VPA ${\mathcal B}$ such that $\languageOf{\mathcal B} = \languageOrderedOf{\grammar}$.
\end{theorem}

By this theorem, given a JSON schema, there exists a VPA $\automaton$ (resp. $\mathcal B$) accepting all JSON documents satisfying this schema (resp.\ only those respecting the key order). These two VPAs can be supposed to be minimal 1-SEVPAs by \Cref{thm:minimal1SEVPA}. The minimal 1-SEVPA $\mathcal B$ could be \emph{exponentially smaller} than the minimal 1-SEVPA $\automaton$ as the order of the key-value pairs is fixed inside objects
(see an illustrating example in \Cref{app:expSmaller}). We use this minimal 1-SEVPA $\mathcal B$ in the next section for the validation of JSON documents.

Given an order on $\keyAlphabet$ and a JSON schema $\schema$ defined by a closed extended CFG $\grammar$, the minimal 1-SEVPA accepting $\languageOrderedOf{\grammar}$ can be learned in the sense of \Cref{thm:learningVPA}. Notice that we have to adapt the learning algorithm if the teacher only knows the schema $\schema$:
\begin{itemize}
    \item A membership query over a JSON document asks whether this document satisfies $\schema$.
    \item An equivalence query is answered\footnote{It is common to proceed this way in automata learning when exact equivalence is intractable, as explained in~\cite[Section 4]{DBLP:journals/iandc/Angluin87}.} by generating a certain number of random (valid and invalid) JSON documents and by verifying that the learned VPA $\automaton[H]$ and the schema $\schema$ agree on the documents' validity.
    \item In membership and equivalence queries, all considered JSON documents have the order of their key-value pairs respecting the given order of $\keyAlphabet$.
\end{itemize}
The randomness used in the equivalence queries implies that the learned 1-SEVPA may not exactly accept $\languageOrderedOf{\grammar}$. Setting the number of generated documents to be large would help reducing the probability that an incorrect 1-SEVPA is learned. In \Cref{sec:implementation}, we discuss in more details how to generate adequate JSON documents for equivalence queries.

Let us now prove \Cref{thm:VPAforJSON}.

\begin{proof}[of \Cref{thm:VPAforJSON}]
Let $\grammar$ be a closed extended CFG defining a JSON schema. To prove the existence of a VPA $\automaton$ accepting $\languageOf{\grammar}$, we adapt a construction provided in the proof of~\cite[Theorem 1]{DBLP:conf/www/KumarMV07}.\footnote{Some difficulties arise from the presence of Boolean operators in the grammar productions.} The terminal alphabet of $\grammar$ is equal to $\JSONalphabet = \callAlphabet \cup \returnAlphabet \cup \internalAlphabet$; suppose that its non-terminal alphabet is equal to $\nonTerminalAlphabet = \{\nonTerminal_0,\nonTerminal_1,\ldots,\nonTerminal_n\}$ where $S_0$ is the axiom. We assume that any production that has $S_0$ as left-hand side is of the form $S_0 \Coloneqq \laccol e \raccol$, where $e$ is a generalized regular expression over the alphabet $\nonTerminalAlphabet \cup \internalAlphabet$, as JSON documents are supposed to be objects. %For instance, in case of productions $S_0 \Coloneqq S_1 \vee S_2$ and $S_i \Coloneqq \laccol e_i \raccol$, $i = 1,2$, then replace the production for $S_0$ by $S_0 \Coloneqq \laccol e_1 \vee e_2 \raccol$.%where $e$ is a generalized regular expression over the alphabet $\nonTerminalAlphabet \cup \internalAlphabet$.\footnote{For instance, in case of productions $S_0 \Coloneqq S_1 \vee S_2$ and $S_i \Coloneqq \laccol e_i \raccol$, $i = 1,2$, then replace the production for $S_0$ by $S_0 \Coloneqq \laccol e_1 \vee e_2 \raccol$.}
%Note that this assumption is no loss of generality as all JSON documents are supposed to be objects. 
We need to apply several steps in a way to construct the required VPA $\automaton$.

\medskip
$(i)$ We transform the grammar $\grammar$ such that: 
\begin{itemize}
    \item the left-hand sides of productions are pairwise distinct,
    \item each production is of the form $\nonTerminal_j \Coloneqq a_j e_j \bar a_j$ such that $a_j \in \{\laccol,\lcrochet\}$ and $e_j$ is a generalized regular expression over the alphabet $\nonTerminalAlphabet \cup \internalAlphabet$.
\end{itemize}
We proceed as follows to obtain the transformed grammar.
\begin{enumerate}
    \item If $\nonTerminal_j$ appears as the left-hand side of $k \geq 2$ productions, then it is replaced by $k$ new non-terminal symbols $\nonTerminal_{j_1},\nonTerminal_{j_2},\ldots,\nonTerminal_{j_k}$ (one for each of the $k$ productions), and each of occurrence of $\nonTerminal_j$ in the right-hand side of productions is replaced by $\nonTerminal_{j_1} \vee \nonTerminal_{j_2} \vee \ldots \vee \nonTerminal_{j_k}$.
    \item Any production $\nonTerminal_j \Coloneqq v$ with $v \in \valAlphabet$ is deleted. Each occurrence of $\nonTerminal_j$ in the right-hand side of productions is replaced by $v$. 
    \item We proceed as in the previous item with all productions $\nonTerminal_j \Coloneqq \nonTerminal_{j_1} \vee \nonTerminal_{j_2} \vee \ldots \vee \nonTerminal_{j_n}$, $\nonTerminal_j \Coloneqq \nonTerminal_{j_1} \wedge \nonTerminal_{j_2} \wedge \ldots \wedge \nonTerminal_{j_n}$, and $\nonTerminal_j \Coloneqq \neg\nonTerminal_{j_1}$.
\end{enumerate}
After this transformation, a non-terminal $\nonTerminal_j$ uniquely identifies a production and the latter is of the form $\nonTerminal_j \Coloneqq a_j e_j \bar a_j$ with $a_j \in \{\laccol,\lcrochet\}$.\footnote{This process terminates because the grammar is well-formed (see Remark~\ref{rem:well-formed}).} Notice that the axiom $\nonTerminal_0$ of the original grammar may have been replaced by several axioms. % (that all belong to $\nonTerminalAlphabet_{\laccol}$).
Given a production $\nonTerminal_j \Coloneqq \laccol e_j \raccol$, the generalized regular expression $e_j$ is of the form $k_1 \phi_1 \comma k_2 \phi_2 \comma \ldots \comma k_n \phi_n$ where each $\phi_i$ is a Boolean expression of elements from $\valAlphabet \cup \nonTerminalAlphabet$. Similarly, the expression $e_j$ in each production $\nonTerminal_j \Coloneqq \lcrochet e_j \rcrochet$ is of the form either $\emptyword \vee \phi_1 (\comma \phi_1)^*$, or $\phi_1 \comma \ldots \comma \phi_1$. 

\medskip
$(ii)$ In a way to obtain a normal form for the productions $\nonTerminal_j \Coloneqq a_j e_j \bar a_j$, we simplify the Boolean expressions $\phi_i$ that appear in the generalized regular expressions $e_j$, and the productions containing them, as follows. 
\begin{enumerate}
    \item Each Boolean expression $\phi_i$ is put into disjunctive normal form (DNF). We will see that $0$ (meaning \emph{false}) may appear inside $\phi_i$ that is thus simplified with the usual Boolean rules. 
    \item In each DNF $\phi_i$, 
    \begin{itemize}
    \item each conjunction that contains some $v, v' \in \valAlphabet$, with $v \neq v'$, is replaced by $0$,
    \item each conjunction that contains some $v \in \valAlphabet$ and $S \in \nonTerminalAlphabet$ is replaced by $0$ (indeed $S$ defines either an object or an array),
    \item each conjunction that contains some $S, S' \in \nonTerminalAlphabet$ is replaced by $0$, whenever $S$ defines an object and $S'$ an array, or the contrary.
    \end{itemize}
    \item Each production $\nonTerminal_j \Coloneqq a_j e_j \bar a_j$ such that $e_j$ contains a Boolean expression $\phi_i$ equal to $0$ is replaced by $\nonTerminal_j \Coloneqq 0$, with one exception detailed hereafter. Each occurrence of $\nonTerminal_j$ in the right-hand side of other productions is thus replaced by $0$. The exception is $\nonTerminal_j \Coloneqq \lcrochet \emptyword \vee \phi_1 (\comma \phi_1)^* \rcrochet$ which is simplified into $\nonTerminal_j \Coloneqq \lcrochet \emptyword \rcrochet$ when $\phi_1$ is equal to $0$.
\end{enumerate}
After this second step, symbol $0$ does not appear in any DNF formula $\phi_i$ and some productions may have the form $\nonTerminal_j \Coloneqq 0$.

In the next steps $(iii)$ and $(iv)$, we want to modify the grammar in a way that operators $\neg$ and $\wedge$ disappear from the DNF formulas $\phi_i$.

\medskip
$(iii)$ The simplified DNF formulas $\phi_i$ have literals of the form $v$, $\neg v$, $\nonTerminal_j$, or $\neg \nonTerminal_j$, with $v \in \keyAlphabet$, $\nonTerminal_j \in \nonTerminalAlphabet$. We want to define each $\neg v$ and $\neg \nonTerminal_j$ by some additional new productions, in a way to replace them by those productions. 

\begin{enumerate}
\item For this purpose, we first add the productions defining the set of all JSON values as given in Lemma~\ref{lem:universal} (we recall that these productions use the non-terminal $\nonTerminalV$).
\item Let us explain how to define $\neg v$, with $v \in \valAlphabet$, by adequate new productions. We need to define all JSON values except $v$. This is possible by adapting the grammar of Lemma~\ref{lem:universal} and by using the non-terminal $\nonTerminalV$:
\begin{itemize}
    \item $\nonTerminalR_1 \Coloneqq v'$  for all  $v' \in \valAlphabet \setminus \{v\}$,
    \item $\nonTerminalR_2 \Coloneqq \laccol k_1 \nonTerminalV \comma k_2 \nonTerminalV\comma \ldots \comma k_n \nonTerminalV \raccol$ for all sequences $(k_1,\ldots,k_n)$, $n \geq 0$, of pairwise distinct keys of $\keyAlphabet$,
    \item $\nonTerminalR_3 \Coloneqq \lcrochet \emptyword \vee \nonTerminalV (\comma \nonTerminalV)^* \rcrochet$,
    \item $\nonTerminal_{\neg v} \Coloneqq \nonTerminalR_1 \vee \nonTerminalR_2 \vee \nonTerminalR_3$.
\end{itemize}
We can then replace each occurrence of $\neg v$ in the DNF formulas $\phi_i$ by $\nonTerminal_{\neg v}$. 
\item Let us now explain how to define $\neg \nonTerminal_j$, with $\nonTerminal_j \in \nonTerminalAlphabet$, by adequate productions. Recall that $\nonTerminal_j$ is the left-hand side of the production $\nonTerminal_j \Coloneqq a_j e_j \bar a_j$. Let us consider the particular example $a_je_j \bar a_j = \laccol k_1 \phi_1 \comma k_2 \phi_2 \raccol$ (the reader could infer the general case from this particular example). We need to define all JSON values except the ones defined by $\nonTerminal_j$. Defining $\neg S_j$ is done thanks to the following grammar: 
\begin{itemize}
    \item $\nonTerminalT_1 \Coloneqq \laccol k'_1 \nonTerminalV \comma k'_2 \nonTerminalV\comma \ldots \comma k'_n \nonTerminalV \raccol$ for all sequences $(k'_1,\ldots,k'_n)$, $n \geq 0$, of pairwise distinct keys of $\keyAlphabet$, except sequence $(k_1,k_2)$,
    \item $\nonTerminalT_1 \Coloneqq \laccol k_1 \neg\phi_1 \comma k_2 \phi_2 \raccol$
    \item $\nonTerminalT_1 \Coloneqq \laccol k_1 \phi_1 \comma k_2 \neg\phi_2 \raccol$
    \item $\nonTerminalT_1 \Coloneqq \laccol k_1 \neg\phi_1 \comma k_2 \neg\phi_2 \raccol$
    \item $\nonTerminalT_2 \Coloneqq v$  for all  $v \in \valAlphabet$,
    \item $\nonTerminalT_3 \Coloneqq \lcrochet \emptyword \vee \nonTerminalV (\comma \nonTerminalV)^* \rcrochet$,
    \item $\nonTerminal_{\neg \nonTerminal_{j}} \Coloneqq \nonTerminalT_1 \vee \nonTerminalT_2 \vee \nonTerminalT_3$.
\end{itemize}
We can then replace each occurrence of $\neg S_j$ in the DNF formulas $\phi_i$ by $\nonTerminal_{\neg \nonTerminal_{j}}$.
\end{enumerate}
With this step, several new productions have been added to the productions $S_j \Coloneqq a_j e_j \bar a_j$, with $S_j \in \nonTerminalAlphabet$. We again apply to those new productions the transformations described in $(i)$ and $(ii)$, and we replace each occurrence of $\neg v$ and $\neg \nonTerminal_j$ by $\nonTerminal_{\neg v}$ and $\nonTerminal_{\neg S_j}$ respectively. Notice there is no occurrence of neither $\neg \nonTerminalV$, nor $\neg \nonTerminalR_j$, nor $\neg \nonTerminalT_j$, $j=1,2,3$, in the new productions\footnote{The used Boolean operators are limited to \(\lor\).}, and therefore no need to define them by some productions. After this step, no production contains the negation of a primitive value or a non-terminal. We use the same notation $\nonTerminalAlphabet$ for the set of all non-terminals.

\medskip
$(iv)$ We now proceed to simplify the conjunctions away.
For each set of non-terminals \(\{S_{i_1}, \dotsc, S_{i_k}\} \subseteq \nonTerminalAlphabet\), we want to define \(\varphi = S_{i_1} \land \dotsc \land S_{i_k}\) by an additional new production with a left-hand side denoted by $\nonTerminal_{\varphi}$.
We have the following cases:
\begin{enumerate}
    \item If the set \(\{S_{i_1}, \dotsc, S_{i_k}\}\) contains some non-terminals  defining objects and some others defining arrays, then we define $\nonTerminal_{\varphi} \Coloneqq 0$.
    \item If the set contains only non-terminals defining objects, we can assume \(S_{i_{\ell}} \Coloneqq \laccol k_1^{\ell} \phi_1^{\ell} \comma \dotso \comma k_{m_{\ell}}^{\ell} \phi_{m_{\ell}}^{\ell} \raccol\).
        We have two possibilities:
        \begin{itemize}
            \item If all $S_{i_\ell}$ use the same sequence $(k_1,\ldots,k_m)$ of keys, then we set $\nonTerminal_\varphi \Coloneqq \laccol k_1 (\wedge_{\ell = 1}^k \phi_1^{\ell}) \comma \dotso \comma k_m (\wedge_{\ell = 1}^k \phi_m^{\ell}) \raccol$. %Notice that for each \(i\), the conjunction \( \wedge_{\ell = 1}^k \phi_i^{\ell}\), once it is put in DNF, is a disjunction of conjunctions defined  created at this step.
            \item Otherwise, we set \(\nonTerminal_\varphi \Coloneqq 0\).
        \end{itemize}
    \item Suppose that the set contains only non-terminals defining arrays.
    \begin{itemize}
    \item In the case where all arrays have a fixed number of elements, we perform as in the previous item. 
    \item If these arrays all define an unbounded number of elements, i.e., \(\nonTerminal_{i_{\ell}} \Coloneqq \lcrochet \emptyword \lor \phi^{\ell} {(\comma \phi^{\ell})}^* \rcrochet\), then we set \(\nonTerminal_\varphi \Coloneqq \lcrochet \emptyword \lor (\wedge_{\ell = 1}^k \phi^{\ell}) {(\comma (\wedge_{\ell = 1}^k \phi^{\ell}))}^* \rcrochet\).
    \item Let us now focus on the case where some non-terminals define a fixed number $m$ of elements (which we can assume is the same for each non-terminal; otherwise, let \(\nonTerminal_\varphi \Coloneqq 0\)), and some other non-terminals define an unbounded number of elements.
    Then, we set \(\nonTerminal_\varphi \Coloneqq \lcrochet (\wedge_{\ell = 1}^k \phi^{\ell}) \comma \dotso \comma (\wedge_{\ell = 1}^k \phi^{\ell}) \rcrochet\), such that we define \(m\) elements.
    \end{itemize}
\end{enumerate}
Once the new productions $\nonTerminal_\varphi$ are defined for all conjunctions \(\varphi = S_{i_1} \land \dotsc \land S_{i_k}\), we replace by $\nonTerminal_\varphi$ every occurrence of $\varphi$ in the DNF formulas of the productions. Moreover, as the new productions with left-hand side $\nonTerminal_\varphi$ may contain occurrences of formulas of the form $\wedge_{\ell = 1}^k \phi^{\ell}$, we put those formulas in DNF and we also replace the resulting conjunctions $\varphi'$ by the non-terminal $\nonTerminal_{\varphi'}$. In this way, no production contains anymore neither negations nor conjunctions.

\medskip
$(v)$ Finally, here is the last step to obtain the required VPA $\automaton$ accepting $\languageOf{\grammar}$. We have a series of productions in our grammar, whose some have a right-hand side equal to $0$. We delete those productions. For the remaining ones, we rename by $\nonTerminalAlphabet = \{\nonTerminal_1,\ldots,\nonTerminal_n\}$ the resulting alphabet of non-terminals, and by $\nonTerminal_j \Coloneqq a_j e_j \bar a_j$ the production whose $\nonTerminal_j \in \nonTerminalAlphabet$ is the unique left-hand side, for all $j \in \{1,\ldots, n\}$.
We partition $\nonTerminalAlphabet$ into $\nonTerminalAlphabet^{\laccol} \cup \nonTerminalAlphabet^{\lcrochet}$, such that $\nonTerminal_j \in \nonTerminalAlphabet^{a_j}$ according to the value of $a_j$ in the production $\nonTerminal_j \Coloneqq a_j e_j \bar a_j$. Recall that the axiom of the original grammar $\grammar$ may have been replaced by several axioms in $\nonTerminalAlphabet$ by step $(i)$. Recall also that the DNF formulas appearing inside each $e_j$ are now disjunctions of symbols from $\nonTerminalAlphabet \cup \internalAlphabet$, that is, $e_j$ is a \emph{classical} regular expression over $\nonTerminalAlphabet \cup \internalAlphabet$ (that is, a regular expression using union, concatenation and Kleene-$*$ operation).\footnote{Steps $(ii)-(iv)$ are new compared to the proof of~\cite[Theorem 1]{DBLP:conf/www/KumarMV07}, due to the presence of Boolean operators $\neg$ and $\wedge$ in the grammar productions.}

It may happen that $\nonTerminalAlphabet$ contains no axiom $\nonTerminal_j$ (since each production $\nonTerminal_j \Coloneqq 0$, with $\nonTerminal_j$ being an axiom, has been deleted). In this case, it is easy to construct a VPA accepting the empty language. For the sequel, we thus suppose that this situation does not hold.

For each production $\nonTerminal_j \Coloneqq a_j e_j \bar a_j$, we construct a complete DFA $\mathcal{B}_j$ over the alphabet $\nonTerminalAlphabet \cup \internalAlphabet$ for the (classical) regular expression $e_j$. Specifically, let $\mathcal{B}_j = (Q_j,q^0_j,F_j,\delta_j)$ where $Q_j$ is the set of states, $q^0_j$ is the initial state, $F_j$ is the set of final states, and $\delta_j$ is the (total) transition function. If this automaton has bin states, we suppose that it is unique and denoted by $\bot_j$.
We then construct the cartesian product $\mathcal B$ of all the automata $\mathcal{B}_j$, $j \in \{1,\ldots, n\}$. 

The required VPA $\automaton$ over $\JSONpushdownAlphabet$ is actually a 1-SEVPA.\@
It has the form $\automaton = (\states, \JSONpushdownAlphabet, \stackAlphabet, \transitionFunction, \{q_0\}, \{q_f\})$ with $q_0 = (q^0_1, \ldots, q^0_j, \ldots, q^0_n)$, $\states = \times_{j=1}^n Q_j \cup \{q_f\}$ where $q_f$ is a new state, $\stackAlphabet = \states \times \{\laccol,\lcrochet\}$, and $\transitionFunction = \callFunction \cup \returnFunction \cup \internalFunction$ defined as follows.
\begin{itemize}
    \item The set of internal transitions $\internalFunction$ is composed of the transitions of $\mathcal{B}$ labeled by symbols in $\internalAlphabet$ (and thus not in $\nonTerminalAlphabet$).
    \item For each transition $(p,\nonTerminal_j, p')$ of $\mathcal B$ labeled by $\nonTerminal_j \in \nonTerminalAlphabet^a$, with $a \in \{\laccol,\lcrochet\}$, we add the following transitions in $\automaton$: 
    \begin{itemize}
        \item the transition $(p,a,q_0,(p,a)) \in \callFunction$ (after reading $a$, we jump to the initial state $q_0$ of $\automaton$ by pushing $(p,a)$ on the stack), and
        \item the transitions $(q,\bar a,(p,a),p') \in \returnFunction$ for all $q = (q_1,\ldots,q_n) \in Q$ such that $q_j$ is a final state in $\mathcal{B}_j$,  (at the end of a run in $\mathcal B$ that is accepting in $\mathcal{B}_j$, with $(p,a)$ on top of the stack of $\automaton$, after reading $\bar a$, we go to the state $p'$ of $\automaton$).
    \end{itemize}
    \item For each $j \in \{1,\ldots,n\}$ such that $\nonTerminal_j$ is an axiom, we add the call transition $(q_0,\laccol,q_0,(q_0,\laccol))$ to $\callFunction$ and the return transitions $(q,\raccol,(q_0,\laccol),q_f)$ to $\returnFunction$ for all $q = (q_1,\ldots,q_n) \in Q$ such that $q_j$ is a final state in $\mathcal{B}_j$.
\end{itemize}
Notice that the constructed 1-SEVPA may be non deterministic by the way the return transitions have been defined (it is deterministic when we only consider the internal and call transitions).

\medskip
Let us prove that $\automaton$ accepts $\languageOf{\grammar}$. To ease the writing of the proof, let us introduce a notation to designate the language of a disjunction over \(\nonTerminalAlphabet\).
Let \(\languageOf{\nonTerminal_{i_1} \lor \dotsb \lor \nonTerminal_{i_n}} = \bigcup_{j = 1}^n \languageOf{\nonTerminal_{i_j}}\), where each \(\nonTerminal_{i_j}\) is a non-terminal of \(\grammar\) and \(\languageOf{\nonTerminal_{i_j}}\) is the set of words that satisfy \(\nonTerminal_{i_j}\).
We also write \(\nonTerminal \in \phi\) when the non-terminal \(\nonTerminal\) is present in the disjunction \(\phi\).

To prove that $\automaton$ accepts $\languageOf{\grammar}$, we want to show the following claim.
\begin{claim}
For any non-terminal \(\nonTerminal_j \in \nonTerminalAlphabet^a\) (with \(a \in \{\laccol, \lcrochet\}\)) and any word \(a w \bar a \in \wellMatched\), it holds that \(a w \bar a \in \languageOf{\nonTerminal_j}\) if, and only if, in $\automaton$, there is some state \(q = (q_1, \dotsc, q_n)\) such that \((q_0, \emptyword) \xrightarrow{w} (q, \emptyword)\) and \(q_{j} \in F_{j}\).
\end{claim}
%That is, we reach a state $q$ in which \(\nonTerminal_j\) can be considered as satisfied.
From that claim, we can derive that if \(\nonTerminal_j\) is an axiom of \(\grammar\), then \(a w \bar a\) satisfies \(\nonTerminal_j\) if and only if the run for \(w\) ends in a state \((q_1, \dotsc, q_n)\) in which \(q_j \in F_j\).
By construction of \(\automaton\), this is equivalent to reach the final state \(q_f\) after reading \(a w \bar a\) from the initial state $q_0$. Hence we conclude that the language of \(\automaton\) is the union of the languages described by each axiom, i.e., \(\automaton\) accepts \(\languageOf{\grammar}\).

Let \(a \in \{\laccol, \lcrochet\}, S_j \in \nonTerminalAlphabet^a\) be a non-terminal of \(\grammar\), and \(a w \bar a \in \wellMatched\).
We prove our claim by induction over the depth of \(w\) (recall that the depth of \(w\) is the maximal number of unmatched call symbols among the prefixes of \(w\)).

\emph{Base case}.
    The depth of \(w\) is zero, i.e., \(w \in \internalAlphabet^*\).
    By construction of \(\automaton\), there exists a run
    \begin{align*}
        &&(q_0, \emptyword) &\xrightarrow{w} (q, \emptyword) & \text{in \(\automaton\)}\\
        &\iff& q_0 &\xrightarrow{w} q & \text{in \(\automaton[B]\)}.
    \end{align*}
    Thus, by construction of \(\automaton[B], w \in \languageOf{\nonTerminal_j}\) if, and only if, \(q_j \in F_j\).

\emph{Induction step}.
    Let \(d \in \N\) and assume the claim holds for every word of depth \(\leq d\).
    Let \(w \in \wellMatched\) of depth \(d + 1\).
    We suppose that the production of \(\nonTerminal_j\) is of the shape \begin{eqnarray} \label{eq}
        \nonTerminal_j &\Coloneqq& a u_1 \phi_1 \dotso u_m \phi_m u_{m + 1} \bar a
    \end{eqnarray} with each \(u_{\ell} \in \internalAlphabet^*\) and \(\phi_j\) is a disjunction over non-terminals of \(\nonTerminalAlphabet\).
    Notice that this covers the definition of objects, and of arrays of shape \(\lcrochet \phi_1 \comma \dotsc \comma \phi_1 \rcrochet\).
    The case for the second definition of arrays (where \(\nonTerminal_j \Coloneqq \lcrochet \emptyword \lor \phi_1 {(\comma \phi_1)}^* \rcrochet\)) can be handled in a similar fashion.
    We prove the equivalence stated in the claim by showing both directions.

\emph{\(\Rightarrow\)}
    Assume \(a w \bar a \in \languageOf{\nonTerminal_j}\).
    We want to prove that there is a run in $\automaton$ reading \(w\) that ends in a state \(q = (q_1, \dotsc, q_n)\) such that \(q_j \in F_j\).

    As \(a w \bar a \in \languageOf{\nonTerminal_j}\), see (\ref{eq}), we can decompose 
    \[w = u_1 b_1 w_1 \bar b_1 \dotso u_m b_m w_m \bar b_m u_{m + 1}\]
    such that \(b_{\ell} w_{\ell} \bar b_{\ell} \in \languageOf{\phi_{\ell}}\) and $w_\ell$ has depth $\leq d$, for each \(\ell \in \{1, \dotsc, m\}\).
    That is, there exists a non-terminal \(\nonTerminal_{h_{\ell}} \in \phi_{\ell}\) such that \(b_{\ell} w_{\ell} \bar b_{\ell} \in \languageOf{\nonTerminal_{h_{\ell}}}\).
    By the induction hypothesis, we thus have states \(r^1, \dotsc, r^m\) such that \((q_0, \emptyword) \xrightarrow{w_{\ell}} (r^{\ell}, \emptyword)\) and \(r^{\ell}_{h_{\ell}} \in F_{h_{\ell}}\) for each \(\ell \in \{1, \dotsc, m\}\).
    Then, we have the following run in \(\automaton\), with \(p^{\ell}, s^{\ell} \in Q\) and \(\alpha_{\ell} \in \stackAlphabet\)
    \begin{align*}
        (q_0, \emptyword)
            &\xrightarrow{u_1}      (p^1, \emptyword)
             \xrightarrow{b_1}      (q_0, \alpha_1)
             \xrightarrow{w_1}      (r^1, \alpha_1)
             \xrightarrow{\bar b_1} (s^1, \emptyword) \xrightarrow{u_2}      \dotso\\
            &\dotso
            (p^m, \emptyword)
            \xrightarrow{b_m}      (q_0, \alpha_m)
             \xrightarrow{w_m}      (r^m, \alpha_m)
             \xrightarrow{\bar b_m} (s^m, \emptyword)
             \xrightarrow{u_{m+1}}  (q, \emptyword)
    \end{align*}
    with the following equivalent run in \(\automaton[B]\)
    \begin{align*}
        q_0
            &\xrightarrow{u_1}                      p^1
             \xrightarrow{\nonTerminal_{h_{1}}}     s^1
             \xrightarrow{u_2}                      \dotso
             \xrightarrow{u_m}                      p^m
             \xrightarrow{\nonTerminal_{h_{m}}}     s^m
             \xrightarrow{u_{m+1}}                  q.
    \end{align*}
    By construction of $\automaton[B]$, as \(a w \bar a \in \languageOf{\nonTerminal_j}\), it must be that \(u_1 \nonTerminal_{h_{1}} u_2 \dotso u_m \nonTerminal_{h_m} u_{m + 1} \in \languageOf{{\automaton[B]}_j}\), which is equivalent to \(q_j \in F_{j}\).

\emph{\(\Leftarrow\)}
    Assume there is a state \(q = (q_1, \dotsc, q_n)\) such that \(q_j \in F_j\) and \((q_0, \emptyword) \xrightarrow{w} (q, \emptyword)\).
    We show that \(a w \bar a \in \languageOf{\nonTerminal_j}\).
    We decompose
    \[w = u'_1 b'_1 w'_1 \bar b'_1 \dotso u'_t b'_t w'_t \bar b'_t u'_{t+1}\]
    with \(t \geq 1\) (as the depth of \(w\) is positive), \(u'_\ell \in \internalAlphabet^*\), and \(b'_\ell w'_\ell \bar b'_\ell \in \wellMatched\) for each $\ell$.
    Thus, we can decompose the given run into
        \begin{align*}
        (q_0, \emptyword)
            &\xrightarrow{u'_1}      (p^1, \emptyword)
             \xrightarrow{b'_1}      (q_0, \alpha_1)
             \xrightarrow{w'_1}      (r^1, \alpha_1)
             \xrightarrow{\bar b'_1} (s^1, \emptyword) \xrightarrow{u'_2}      \dotso\\
            &\dotso
            (p^t, \emptyword)
            \xrightarrow{b'_t}      (q_0, \alpha_t)
             \xrightarrow{w'_t}      (r^t, \alpha_t)
             \xrightarrow{\bar b'_t} (s^t, \emptyword)
             \xrightarrow{u'_{t+1}}  (q, \emptyword).
    \end{align*}

As \(q_j \in F_j\), it must be that \(q_j \neq \bot_j\). By construction of \(\automaton[A]\) and $\automaton[B]$, we then have \(p^\ell_j \neq \bot_j\) and \(s^\ell_j \neq \bot_j\) for every \(\ell \in \{1, \dotsc, t\}\). %, and we have the corresponding run in $\automaton[B]$     
    Moreover, for each \(\ell \in \{1,\ldots,t\}\), there must exist a non-terminal \(\nonTerminal_{h_{\ell}}\) such that \(r^{\ell}_{h_{\ell}} \in F_{h_{\ell}}\).
    We thus have a corresponding run in \(\automaton[B]\)
        \begin{align*}
        q_0
            &\xrightarrow{u'_1}                      p^1
             \xrightarrow{\nonTerminal_{h_{1}}}     s^1
             \xrightarrow{u'_2}                      \dotso
             \xrightarrow{u'_t}                      p^t
             \xrightarrow{\nonTerminal_{h_{t}}}     s^t
             \xrightarrow{u'_{t+1}}                  q.
    \end{align*}
    
    By the induction hypothesis, it follows that \(b'_{\ell} w'_{\ell} \bar b'_{\ell} \in \languageOf{\nonTerminal_{h_{\ell}}}\) for all $\ell$.
    Recall that, by our assumption, the latter run, seen in $\automaton[B]_j$, leads to the final state $q_j \in F_j$.
    That is, \(u'_1 \nonTerminal_{h_1} \dotso u'_t \nonTerminal_{h_t} u'_{t + 1}\) belongs to \(\languageOf{\automaton[B]_j}\).
    By~\eqref{eq}, this means that \(t = m\), each $u'_\ell$ is equal to $u_\ell$, and each $S_{h_\ell}$ belongs to $\phi_\ell$.
    Therefore $a w \bar a$ belongs to $\languageOf{\nonTerminal_j}$.
    The claim is thus proved. As stated before, we then deduce that \(\automaton\) accepts \(\languageOf{\grammar}\).

\medskip We have thus proved the first statement of Theorem~\ref{thm:VPAforJSON}. This theorem also states that given a grammar $\grammar$ and an order $<$ over $\keyAlphabet$, a VPA $\automaton$ can be constructed such that $\languageOf{\automaton} = \languageOrderedOf{\grammar}$. The construction is exactly the same as done above, by taking into account the fixed order imposed on $\keyAlphabet$ in the productions defining objects.
\qed\end{proof}

\begin{comment}
----------------------------------------------------------

\begin{proof}[of \Cref{thm:VPAforJSON}]
Let $\grammar$ be a closed extended CFG defining a JSON schema. Let us prove the existence of a VPA $\automaton$ such that $\languageOf{\automaton} = \languageOf{\grammar}$. Given $\grammar$, we first construct an automaton $\automaton'$ exactly as in the proof of \Cref{lem:grammarToVPA}. The only difference is that the negation can now appear in the generalized regular expressions $e_i$ of the productions $\schema_i \Coloneqq a_i e_i \bar a_i$ of $\grammar$. To obtain $\automaton$, we then intersect $\automaton'$ with the universal automaton $\automatonUniv$ of \Cref{cor:universal}, by using \Cref{thm:closeness}. 

Finally, a VPA $\mathcal{B}$ such that $\languageOf{\mathcal{B}} =  \languageOrderedOf{\grammar}$ is obtained by taking the intersection of the VPA $\automaton'$ with the VPA $\automatonOrdUniv$ of \Cref{cor:universal}. 
\qed\end{proof}

------------------------------------------
\end{comment}

While this proves that there exists a VPA accepting the same set of documents as a JSON schema, this is not a constructive approach.
Indeed, we do not explain how to go from a schema to a grammar, and this is not trivial due to how Boolean operations are defined inside a schema.
Moreover, to get this automaton as explained in the proof of Theorem~\ref{thm:VPAforJSON}, the given grammar is modified through several steps that may lead to a huge number of productions, hence to a huge 1-SEVPA.\@
Therefore, we rely on a learning algorithm to build a 1-SEVPA (see \Cref{sec:definitions:learning}).
The learning process itself also has the advantage of being agnostic of the semantics of a schema which is still changing and being debated: only the classical validator needs a version of the semantics.

\begin{comment}
\textcolor{blue}{While this proves that there exists a VPA accepting the same set of documents as a JSON schema, this is not a constructive approach.
Indeed, we do not explain how to go from a schema to a grammar, and this is not trivial due to how Boolean operations are defined inside a schema.
Moreover, to ease the following proofs and algorithms, we want to construct a 1-SEVPA rather than any VPA\@.
Therefore, we rely on a learning algorithm to build a 1-SEVPA (see \Cref{sec:definitions:learning}).
The learning process itself also has the advantage of being agnostic of the semantics of a schema which are still changing and being debated: only the classical validator needs a version of the semantics.
}
\end{comment}

\section{Validation of JSON Documents}\label{sec:validation}
%====================================

One important computational problem related to JSON schemas consists in determining whether a JSON document~$J$ satisfies a schema $\schema$.
In this section, we provide a \emph{streaming} algorithm that validates a JSON document against a JSON schema.
By \enquote{streaming}, we mean an algorithm that performs the validation test by processing the document in a single pass, symbol by symbol, and by using a limited amount of memory with respect to the size of the given document.

Our approach is new and works as follows. Given a JSON schema $\schema$, we learn the minimal 1-SEVPA $\automaton$ accepting the language $\languageOf{\automaton}$ equal to the set of all JSON documents $J$ satisfying $S$ and respecting a given order $<$ on $\keyAlphabet$. We know that this is possible as explained in the previous section. Unfortunately, checking whether a JSON document $J$ satisfies the JSON schema $\schema$ does not amount to checking whether $J \in \languageOf{\automaton}$ as the key-value pairs inside the objects of $J$ can be arbitrarily ordered. Instead, we design a streaming algorithm that uses $\automaton$ in a clever way to allow arbitrary orders of key-value pairs. To do this, we use a \emph{key graph} defined and shown to be computable in the sequel.
Then, we describe our validation algorithm and study its complexity.

Henceforth we fix a schema $\schema$ given by a closed extended CFG $\grammar$, an order $<$ on $\keyAlphabet$, and a 1-SEVPA $\automaton = (\states, \JSONpushdownAlphabet, \stackAlphabet, \transitionFunction, \{q_0\}, \finalStates)$ accepting $\languageOrderedOf{\grammar}$.

\subsection{Key Graph}\label{subsec:KeyGraph}
%========================================

In this section, w.l.o.g.\ we suppose that $\automaton$ has \emph{no bin states}. 
Let $\transitionSystem$ be the transition system of the 1-SEVPA $\automaton$.
We explain how to associate to $\automaton$ a particular graph $\keyGraph$, called \emph{key graph}, abstracting the paths of $\transitionSystem$ labeled by the contents of the objects appearing in words of $\languageOrderedOf{\grammar}$.
This graph is essential in our validation algorithm.

\begin{definition}\label{def:keyGraph} 
The \emph{key graph} $\keyGraph$ of $\automaton$ has:
\begin{itemize}
\item the vertices $(p,k,p')$ with $p,p' \in \states$ and $k \in \keyAlphabet$ if there exists in $\transitionSystem$ a path $\langle p,\emptyword \rangle \xrightarrow{kv} \langle p',\emptyword \rangle$ with $v \in \valAlphabet \cup \{a u \bar a \mid a \in \callAlphabet, u \in \wellMatched[\JSONpushdownAlphabet]\}$,\footnote{Notice that each vertex $(p,k,p')$ of $\keyGraph$ only stores the key $k$ and not the word $kv$.}
\item the edges $((p_1,k_1,p'_1),(p_2,k_2,p'_2))$ if there exists $(p'_1,\comma,p_2) \in \internalFunction$.
\end{itemize}
\end{definition}

We have the following property. 
\begin{lemma}\label{lem:abstractKeyGraph}
There exists a path $((p_1,k_1,p'_1)(p_2,k_2,p'_2)\ldots (p_n,k_n,p'_n))$ in $\keyGraph$ with $p_1 = q_0$ if and only if there exist 
\begin{itemize}
    \item a word $u = k_1 v_1 \comma k_2 v_2 \comma \ldots \comma k_n v_n$ such that each $k_iv_i$ is a key-value pair and $u$ is a factor of a word in $\languageOrderedOf{\grammar}$, 
    \item and a path $\langle q_0, \emptyword \rangle \xrightarrow{u} \langle p'_n, \emptyword \rangle$ in $\transitionSystem$ that decomposes as follows:
    \begin{gather*}
        \langle p_i,\emptyword \rangle \xrightarrow{k_iv_i} \langle p'_i,\emptyword \rangle, \forall i \in \{1,\ldots,n\}\\
        \shortintertext{ and }
        \langle p'_i,\emptyword \rangle \xrightarrow{\comma} \langle p_{i+1},\emptyword \rangle, \forall i \in \{1,\ldots,n-1\}.
    \end{gather*}
\end{itemize}
\end{lemma}

\noindent
Hence, paths in $\keyGraph$ focus on contents of objects being part of JSON documents satisfying $\schema$. Moreover, they abstract paths in $\transitionSystem$ in the sense that only keys $k_i$ are stored and the subpaths labeled by the values $v_i$ are implicit.

\begin{example}
    \begin{figure}
        \centering
        \begin{tikzpicture}[
    automaton,
    node distance=1.1cm and 3cm,
]
    \node [state, initial]                  (q0)    {\(q_0\)};
    \node [state, above=of q0]              (q1)    {\(q_1\)};
    \node [state, right=of q1]              (q2)    {\(q_2\)};
    \node [state, right=of q2]              (q3)    {\(q_3\)};
    \node [state, right=of q3]              (q4)    {\(q_4\)};
    % \node [state, right=of q0]              (q5)    {\(q_5\)};
    % \node [state, above right=of q5]        (q6)    {\(q_6\)};
    % \node [state, below right=of q5]        (q7)    {\(q_7\)};
    % \node [state, right=of q7]              (q8)    {\(q_8\)};
    % \node [state, right=of q8]              (q9)    {\(q_9\)};
    \node [state, right=of q0]              (q10)   {\(q_{5}\)};
    \node [state, right=of q10]             (q11)   {\(q_{6}\)};
    \node [state, right=of q11]             (q12)   {\(q_{7}\)};
    \node [state, below=of q12]             (q13)   {\(q_{8}\)};
    \node [state, left=of q13]              (q14)   {\(q_{9}\)};
    \node [state, left=of q14]              (q15)   {\(q_{10}\)};
    \node [state, left=of q15, accepting]   (q16)   {\(q_{11}\)};

    \foreach \start/\lbl/\target/\option in {
        q0/\texttt{title}/q1/,%
        q1/\texttt{s}/q2/,%
        q2/\(\comma\)/q3/,%
        q3/\texttt{conference}/q4/,%
        % q0/s/q5,%
        % q7/\(\comma\)/q8,%
        % q8/"conference":/q9,%
        q0/\texttt{name}/q10/,%
        q10/\texttt{s}/q11/,%
        q11/\(\comma\)/q12/,%
        q12/\texttt{year}/q13/,%
        q13/\texttt{i}/q14/'%
    }
    {
        \path (\start)  edge [\option]  node {\lbl} (\target);
    }

    \path
        % (q5)    edge [bend left]    node {\(\comma\)}   (q6)
        % (q6)    edge [bend left]    node {s}            (q5)
        % (q5)    edge                node {\(\rcrochet, (q_4, \lcrochet)\)}  (q7)
        (q14)   edge [']    node {\(\raccol, (q_4, \laccol)\)}  (q15)
        (q15)   edge [']    node {\(\raccol, (q_0, \laccol)\)}  (q16)
    ;
\end{tikzpicture}
        \caption{A 1-SEVPA for the schema from \Cref{fig:json_schema}, without considering the key \texttt{keywords}.}%
        \label{fig:vpa_for_json}

\vspace{.5cm}

        \begin{tikzpicture}[
    automaton,
    node distance=1cm and 5cm,
]
    \node [state]                       (q0 title q2)       {\(q_0, \texttt{title}, q_2\)};
    \node [state, right=of q0 title q2] (q3 conference q10) {\(q_3, \texttt{conference}, q_{10}\)};
    \node [state, below=of q0 title q2] (q0 name q6)        {\(q_0, \texttt{name}, q_6\)};
    \node [state, right=of q0 name q6]  (q7 year q9)        {\(q_7, \texttt{year}, q_9\)};

    \path
        (q0 title q2)   edge    (q3 conference q10)
        (q0 name q6)    edge    (q7 year q9)
    ;
\end{tikzpicture}
        \caption{The key graph for the 1-SEVPA from \Cref{fig:vpa_for_json}.}%
        \label{fig:key_graph_for_json}
    \end{figure}

    Consider the schema from \Cref{fig:json_schema}, without the key \verb!keywords! and its associated schema. It is defined by the following closed CFG $\grammar$ (where the related productions with key permutations are not indicated, see Example~\ref{ex:abstractJSONschema}): 
    \begin{eqnarray*}
\nonTerminal_0 &\Coloneqq& \laccol \texttt{title} ~ \nonTerminal_1 \comma \texttt{conference} ~\nonTerminal_2 \raccol \\
\nonTerminal_1 &\Coloneqq&  \verb!s!\\
\nonTerminal_2 &\Coloneqq& \laccol \texttt{name} ~ \nonTerminal_1 \comma \texttt{year} ~\nonTerminal_3 \raccol \\
\nonTerminal_3 &\Coloneqq&  \verb!i! 
\end{eqnarray*}

    A 1-SEVPA $\automaton$ accepting $\languageOrderedOf{\grammar}$ is given in \Cref{fig:vpa_for_json}. For clarity, call transitions\footnote{Recall the particular form of call transitions and stack alphabet for 1-SEVPAs, see Definition~\ref{def:1SEVPA}.} and the bin state are not represented.
    In Figure~\ref{fig:key_graph_for_json}, we depict its corresponding key graph \(\keyGraph\).
    Since we have the path \(\langle q_0, \emptyword \rangle \xrightarrow{\texttt{title} ~ \texttt{s}} \langle q_2, \emptyword \rangle\) in $\transitionSystem$, the triplet \((q_0, \texttt{title}, q_2)\) is a vertex of \(\keyGraph\).
    Likewise, \((q_0, \texttt{name}, q_6)\) and \((q_7, \texttt{year}, q_9)\) are vertices.
    Since we have the path 
    \[
        \langle q_4, \emptyword \rangle \xrightarrow{\laccol} \langle q_0, (q_4, \laccol) \rangle   \xrightarrow{\texttt{name} ~ \texttt{s} ~\comma~ \texttt{year} ~ \texttt{i}} \langle q_9, (q_4, \laccol) \rangle  \xrightarrow{\raccol} \langle q_{10}, \emptyword \rangle,
    \]
    the triplet \((q_3, \texttt{conference}, q_{10})\) is also a vertex of $\keyGraph$. Finally, since \(\langle q_2, \emptyword \rangle \xrightarrow{\comma} \langle q_3, \emptyword \rangle\), we have an edge from \((q_0, \texttt{title}, q_2)\) to \((q_3, \texttt{conference}, q_{10})\).
\end{example}

\begin{proof}[of \Cref{lem:abstractKeyGraph}]
We only prove one implication, the other being easily proved. Suppose that there exists a path $((p_1,k_1,p'_1)(p_2,k_2,p'_2)\ldots (p_n,k_n,p'_n))$ in $\keyGraph$ with $p_1 = q_0$. Then by definition of $\keyGraph$, there exists in $\transitionSystem$ a path 
\begin{eqnarray}\label{eq:path}
\langle q_0, \emptyword \rangle \xrightarrow{u} \langle p'_n, \emptyword \rangle    
\end{eqnarray}
with $u = k_1 v_1 \comma k_2 v_2 \comma \ldots \comma k_n v_n$ such that $v_i \in \valAlphabet \cup \{a u \bar a \mid a \in \callAlphabet, u \in \wellMatched[\JSONpushdownAlphabet]\}$ for all $i$.
As $p'_n$ is not a bin state, there exists another path in $\transitionSystem$
\begin{eqnarray} \label{eq:noBinState}
    \left\langle q_0, \emptyword \right\rangle \xrightarrow{t} \left\langle p'_n, \stack \right\rangle \xrightarrow{t'} \left\langle q, \emptyword \right\rangle
\end{eqnarray}
with $q \in \finalStates$. Let us decompose $t$ in terms of its unmatched call symbols, that is, $t = t_1a_1t_2a_2 \ldots t_m a_m t_{m+1}$ such that $m \geq 0$ and $t_i \in \wellMatched[\JSONpushdownAlphabet], a_i \in \callAlphabet$ for all $i$. It follows that in~\eqref{eq:noBinState}, $|\stack| = m$. Therefore, with $t = t'' t_{m+1}$ and since $\automaton$ is a 1-SEVPA, we get the path $\left\langle q_0, \emptyword \right\rangle \xrightarrow{t''} \left\langle q_0, \stack \right\rangle \xrightarrow{t_{m+1}} \left\langle p'_n, \stack \right\rangle \xrightarrow{t'} \left\langle q, \emptyword \right\rangle$. In the latter path, by (\ref{eq:path}), we can replace $t_{m+1}$ by $u$ showing that the word $t''ut'$ belongs to $\languageOf{\automaton} = \languageOrderedOf{\grammar}$. Moreover, as each $v_i \in \valAlphabet \cup \{a u \bar a \mid a \in \callAlphabet, u \in \wellMatched[\JSONpushdownAlphabet]\}$, it follows that $k_iv_i$ is a key-value pair for all $i$.
\qed\end{proof}

From \Cref{lem:abstractKeyGraph}, we get that the key graph contains a finite number of paths.
%as indicated in the next corollary. 
\begin{corollary}\label{cor:acyclic}
In the key graph 
$\keyGraph$, there is no path $((p_1,k_1,p'_1)%(p_2,k_2,p'_2)
\ldots (p_n,k_n,p'_n))$ with $p_1 = q_0$ such that $k_i = k_j$ for some $i \neq j$.
\end{corollary}

\begin{proof}
Assume the opposite. By \Cref{lem:abstractKeyGraph}, there exists a path $\langle q_0, \emptyword \rangle \xrightarrow{u} \langle p'_n, \emptyword \rangle$ in $\transitionSystem$ such that $u = k_1 v_1 \comma k_2 v_2 \comma \ldots \comma k_n v_n$ is factor of a word in $\languageOrderedOf{\grammar}$. This is impossible because keys must be pairwise distinct inside objects appearing in JSON documents.
\qed\end{proof}

\begin{lemma}\label{lem:sizeKeyGraph}
The key graph $\keyGraph$ has $\complexity(\automatonSize^2 \cdot |\keyAlphabet|)$ vertices. Moreover, visiting all vertices along all the paths of $\keyGraph$ that start from a vertex $(p,k,p')$ such that $p = q_0$ is in $\complexity(|Q \times \keyAlphabet|^{|\keyAlphabet|})$.
\end{lemma} 

\begin{proof}
The first statement is trivial. Let us prove the second one. First, notice that a given vertex $(p,k,p')$ has at most $\alpha = |\keyAlphabet \times Q|$ successors $(q,k',q')$ as $\automaton$ is deterministic. Second, the length of a longest path in $\keyGraph$ is bounded by $|\keyAlphabet|$ by \Cref{cor:acyclic}. Third, the number of vertices in the longest paths starting with the vertex $(q_0,k,p')$ is bounded by $\frac{\alpha^{|\keyAlphabet|} - 1}{\alpha - 1}$. As there are $\alpha$ potential starting vertices $(q_0,k,p')$, the announced upper bound follows. 
\qed\end{proof}

\subsection{Computation of the Key Graph}\label{subsec:KeyGraphComputation}
%============================================================================

In this section, we provide a polynomial time algorithm to compute $\keyGraph$ from $\automaton$. First, we explain how to compute the reachability relation $\accRelation \subseteq Q^2$:

\begin{itemize}
    \item Initially, $\accRelation$ is the transitive closure of $\identity{\states} \cup \{(q,q') \mid \exists (q,a,q') \in \internalFunction \}$.
    \item Repeat until $\accRelation$ does not change:
    \begin{itemize}
        \item Add to $\accRelation$ all elements $(q,q')$ such that there exist $(q,a,p,\gamma) \in \callFunction$, $(p,p') \in \accRelation$, and $(p',\bar a,\gamma,q') \in \returnFunction$,
        \item Close $\accRelation$ with its transitive closure.
    \end{itemize}
\end{itemize}
This process terminates as the set $\states^2$ is finite. It is easy to see that $\accRelation$ can be computed in time polynomial in $|\states|$ and $|\transitionFunction|$.

Second, recall that $\automaton$ can have bin states. We compute them, if there are any, and we remove them and the related transitions from $\automaton$. A polynomial time algorithm detecting the bin states is given in \Cref{app:ComplexityKeyGraph}.

Finally, let us give an algorithm computing $\keyGraph$. Its vertices are computed as follows: $(p,k,p')$ is a vertex in $\keyGraph$ if there exist $(p,k,q) \in \internalFunction$ with $k \in \keyAlphabet$ and 
\begin{itemize}
    \item $(q,a,p') \in \internalFunction$  with $a \in \valAlphabet$,
    \item or $(q,a,r,\gamma) \in \callFunction$, $(r,r') \in \accRelation$, and $(r',\bar a,\gamma,p') \in \returnFunction$.
\end{itemize}
We compute the edges of $\keyGraph$ as follows: $((p_1,k_1,p'_1),(p_2,k_2,p'_2))$ is an edge in $\keyGraph$ if there exists $(p'_1,\comma,p_2) \in \internalFunction$.

Therefore, the given algorithm for computing $\keyGraph$ is polynomial, whose precise complexity is given in the next proposition and proved in \Cref{app:ComplexityKeyGraph}.

\begin{proposition}\label{prop:KeyGraph}
Computing the key graph $\keyGraph$ is in time $\complexity(|\transitionFunction|^2 + |\transitionFunction|\cdot |\states|^4 + |\states|^5 + |\states|^4\cdot|\keyAlphabet|^2)$.
\end{proposition}

\subsection{Validation Algorithm}\label{subsec:validationAlgo}
%===============================

In this section, we provide a streaming algorithm that validates JSON documents against a given JSON schema. 

Given a word $w \in \JSONalphabet^* \setminus \{\emptyword\}$, we want to check whether $w \in \languageOf{\grammar}$. The main difficulty is that the key-value pairs inside an object are arbitrarily ordered in $w$ while a fixed key order is encoded in the 1-SEVPA $\automaton$ ($\languageOf{\automaton} = \languageOrderedOf{\grammar}$). Our validation algorithm is inspired by the algorithm computing a det-VPA equivalent to some given VPA~\cite{DBLP:conf/stoc/AlurM04} (see \Cref{thm:detVPA} and its proof) and uses the key graph $\keyGraph$ to treat arbitrary orders of the key-value pairs inside objects. 

During the reading of $w \in \JSONalphabet^*\setminus \{\emptyword\}$, in addition to checking whether $w \in \wellMatched[\JSONpushdownAlphabet]$, the algorithm updates a subset $\setOfStates \subseteq \accRelation$ and modifies the content of a stack $\St$ (push, pop, modify the element on top of $\St$). 

First, let us explain the information stored in $\setOfStates$.  Assume that we have read the prefix $zau$ of $w$ such that $a \in \callAlphabet$ is the last unmatched call symbol (thus $za \in (\wellMatched[\JSONpushdownAlphabet] \cdot \callAlphabet)^*$ and $u \in \wellMatched[\JSONpushdownAlphabet]$). 
\begin{itemize}
\item If $a$ is the symbol $\lcrochet$, then we have $\setOfStates = \{(p,q) \mid \langle p,\emptyword \rangle \xrightarrow{u} \langle q,\emptyword \rangle \}$.
\item If $a$ is the symbol $\laccol$, then we have $u = k_1 v_1 \comma k_2 v_2 \comma \ldots k_{n-1} v_{n-1} \comma u'$ such that $u' \in \wellMatched[\JSONpushdownAlphabet]$ and $u'$ is prefix of $k_n v_n$, where each $k_iv_i$ is a key-value pair. Then $\setOfStates = \{(p,q) \mid \langle p,\emptyword \rangle \xrightarrow{u'} \langle q,\emptyword \rangle \}$.
\end{itemize}
In the first case, by using $\setOfStates$ as defined previously, we adopt the same approach as for the determinization of VPAs. In the second case, with $u$, we are currently reading the key-value pairs of an object in some order that is not necessarily the one encoded in $\automaton$. In this case the set $\setOfStates$ is focused on the currently read key-value pair $k_nv_n$, that is, on the word $u'$ (instead of the whole word $u$). After reading of the whole object $\laccol k_1 v_1 \comma k_2 v_2 \comma \ldots \raccol$, we will use the key graph $\keyGraph$ to update the current set $\setOfStates$.   

Second, let us explain the form of the elements stored in the stack $\St$. Such an element is
%\begin{itemize}
    %\item 
    either a pair $(\setOfStates,\lcrochet)$,
    %\item 
    or a 5-tuple $(\setOfStates, \laccol, K, k, \Marks)$,
%\end{itemize}
where $\setOfStates$ is a set as described previously, $K \subseteq \keyAlphabet$ is a subset of keys, $k \in \keyAlphabet$ is a key, and $\Marks$ is a set containing some vertices of $\keyGraph$.\footnote{In the particular case of the object $\laccol \raccol$, the 5-tuple $(\setOfStates, \laccol, K, k, \Marks)$ is replaced by $(\setOfStates,\laccol)$. This situation will be clarified during the presentation of our algorithm.}

We are now ready to detail our streaming validation algorithm.\footnote{We focus on generic VPAs first, and then specialize to 1-SEVPAs.} Before beginning to read $w$, we initialize $\setOfStates$ to the set $\identity{\{q_0\}}$ and $\St$ to the empty stack. We are now going to explain how to update the current set $\setOfStates$ and the current contents of the stack $\St$ while reading the input word $w$.
Suppose that we are reading the symbol $a$ in $w$. In some cases we will also peek the symbol $b$ following $a$ (\emph{lookahead} of one symbol). 
\begin{description}
    \item[Case (1)] Suppose that $a$ is the symbol $\lcrochet$. This means that we begin to read an array. Hence $(\setOfStates,\lcrochet)$ is pushed on $\St$ and $\setOfStates$ is updated to $\Rupdate = \identity{\states}$. We thus proceed exactly as in the proof of \Cref{thm:detVPA}.

    \item[Case (2)] Suppose that $a \in \internalAlphabet$ and $\lcrochet$ appears on top of $\St$. We are thus reading the elements of an array. Hence $\setOfStates$ is updated to $\Rupdate = \{(p,q) \mid \exists (p,q') \in \setOfStates, (q',a,q) \in \internalFunction \}$. Again we proceed as in the proof of \Cref{thm:detVPA}.

    \item[Case (3)] Suppose that $a$ is the symbol $\rcrochet$. This means that we finished reading an array. If the stack is empty or its top element contains $\laccol$, then $w \not\in \languageOf{\grammar}$ and we stop the algorithm. Otherwise $(\setOfStates',\lcrochet)$ is popped from $\St$ and $\setOfStates$ is updated to \[\Rupdate = \{(p,q) \mid \exists (p,p') \in \setOfStates', (p',\lcrochet,r',\gamma) \in \callFunction, (r',r) \in \setOfStates, (r,\rcrochet,\gamma,q) \in \returnFunction \},\]
    as in the proof of \Cref{thm:detVPA}. 

    \item[Case (4)] Suppose that $a$ is the symbol $\laccol$. 
    \begin{itemize}
        \item Let us first consider the particular case where the symbol $b$ following $\laccol$ is equal to $\raccol$, meaning that we will read the object $\laccol \raccol$. In this case, $(\setOfStates,\laccol)$ is pushed on $\St$ and $\setOfStates$ is updated to $\Rupdate = \identity{\states}$ as in Case (1). 
        \item Otherwise, if $b$ belongs to $\keyAlphabet$,  we begin to read a (non-empty) object whose treatment is different from that of an array as its key-value pairs can be read in any order. Then, $\setOfStates$ is updated to $\Rupdate = \identity{P_b}$ where 
        \[P_b = \{p \in Q \mid \exists (p,b,p') \in \keyGraph \},\] 
        and $(\setOfStates,\laccol,K,b,\Marks)$ is pushed on $\St$ such that $K$ is the singleton $\{b\}$ and $\Marks$ is the empty set. The 5-tuple pushed on $\St$ indicates that the key-value pair that will be read next begins with key $b$; moreover $K = \{b\}$ because this is the first pair of the object. The meaning of $\Marks$ will be clarified later. The updated set $\Rupdate$ is equal to the identity relation on $P_b$ since after reading $\laccol$, we will start reading a key-value pair whose abstracted state in $\keyGraph$ can be any state from $P_b$. Later while reading the object whose reading is here started, we will update the 5-tuple on top of $\St$ as explained below.
        \item Finally, it remains to consider the case where $b \not \in \keyAlphabet \cup \{\raccol\}$. In this final case, we have that $w \not\in \languageOf{\grammar}$ and we stop the algorithm.
    \end{itemize}
    
    \item[Case (5)] Suppose that $a \in \internalAlphabet \setminus \{\comma\}$ and $\laccol$ appears on top of $\St$. Therefore we are currently reading a key-value pair of an object. Then $\setOfStates$ is updated to $\Rupdate = \{(p,q) \mid \exists (p,q') \in \setOfStates, (q',a,q) \in \internalFunction \}$.

    \item[Case (6)] Suppose that $a$ is the symbol $\comma$ and $\laccol$ appears on top of $\St$. This means that we just finished reading a key-value pair whose key $k$ is stored in the 5-tuple $(\setOfStates',\laccol,K,k,\Marks)$ on top of $\St$, and that another key-value pair will be read after symbol $\comma$. The set $K$ in $(\setOfStates',\laccol,K,k,\Marks)$ stores all the keys of the key-values pairs already read including $k$.
    
    \begin{itemize}
        \item If the symbol $b$ following $\comma$ does not belong to $\keyAlphabet$, then $w \not\in \languageOf{\grammar}$ and we stop the algorithm. 
        \item Otherwise, if $b$ belongs to $K$, this means that the object contains twice the same key, that is, $w \not\in \languageOf{\grammar}$, and we also stop the algorithm. 
        \item Otherwise, the set $\setOfStates$ is updated to $\Rupdate = \identity{P_b}$ (as we begin the reading of a new key-value pair whose key is $b$) and the 5-tuple $(\setOfStates',\laccol,K,k,\Marks)$ on top of $\St$ is updated such that 
        \begin{itemize}
            \item $K$ is replaced by $K \cup \{b\}$,
            \item $k$ is replaced by $b$, and
            \item all vertices $(p,k,p')$ of $\keyGraph$ such that $(p,p') \not\in \setOfStates$ are added to the set $\Marks$.
        \end{itemize}
    Recall that the vertex $(p,k,p')$ of $\keyGraph$ is a witness of a path $\langle p,\emptyword \rangle \xrightarrow{kv} \langle p',\emptyword \rangle$ in $\transitionSystem$ for some key-value pair $kv$. Hence by adding this vertex $(p,k,p')$ to $\Marks$, we mean that the pair that has just been read does not use such a path.
    \end{itemize}

    \item[Case (7)] Suppose that $a$ is the symbol $\raccol$. Therefore we end the reading of an object. If the stack is empty or its top element contains $\lcrochet$, then $w \not\in \languageOf{\grammar}$ and we stop the algorithm. Otherwise the top of $\St$ contains either $(\setOfStates',\laccol)$ or $(\setOfStates',\laccol,K,k,\Marks)$ that we pop from $\St$. 
    
    \begin{itemize}
    \item If $(\setOfStates',\laccol)$ is popped, then we are ending the reading of the object $\laccol \raccol$. Hence, we proceed as in Case (3): $\setOfStates$ is updated to 
    \[\Rupdate = \{(p,q) \mid \exists (p,p') \in \setOfStates', (p',\laccol,r',\gamma) \in \callFunction, (r',\raccol,\gamma,q) \in \returnFunction \}.\footnote{Notice that $\setOfStates$ does not appear in $\Rupdate$ as $\setOfStates = \identity{Q}$.}\]
    \item If $(\setOfStates',\laccol,K,k,\Marks)$ is popped, we are reading the last key-value pair with key $k$ and we add to $\Marks$ all vertices $(p,k,p')$ of $\keyGraph$ such that $(p,p') \not\in \setOfStates$ as done in Case (6). Next, let $\validPaths(K,\Marks)$ be the set of pairs of states $(r,r') \in \states^2$ such that there exists a path $((p_1,k_1,p'_1) \ldots (p_n,k_n,p'_n))$ in $\keyGraph$ with
    \begin{itemize}
        \item $p_1 = r$, $p'_n = r'$,
        \item $(p_i,k_i,p'_i) \not\in \Marks$ for all $i \in \{1, \ldots, n\}$,
        \item $K = \{k_1,k_2,\ldots,k_n\}$.
    \end{itemize}
    Then $\setOfStates$ is updated to:
    \begin{align*}
        \Rupdate = \{
            (p,q) \mid
                &\exists (p,p') \in \setOfStates',
                (p',\laccol,r',\gamma) \in \callFunction,\\
                &(r',r) \in \validPaths(K,\Marks),
                (r,\raccol,\gamma,q) \in \returnFunction
        \}.
    \end{align*}
    We thus proceed as in Case (3) except that condition $(r',r) \in \setOfStates$ is replaced by $(r',r) \in \validPaths(K,\Marks)$. With the latter condition, we check that the key-value pairs that have been read as composing an object of $w$ are such that the same pairs ordered as imposed by $\automaton$ label some path in $\transitionSystem$, that is, the corresponding abstracted path appears in $\keyGraph$. 
    \end{itemize}
    \item[Case (8)] Suppose that $a \in \internalAlphabet$ and $\St$ is empty, then $w \not\in \languageOf{\grammar}$ and we stop the algorithm. Indeed an internal symbol appears either in an array or in an object (see Cases (2), (5), and (6) above). 
\end{description}

Finally, when the input word $w$ is completely read, we check whether the stack $\St$ is empty and the computed set $\setOfStates$ contains a pair $(q_0,q)$ with $q \in \finalStates$. 

Notice that the previous algorithm is supposed to have a 1-SEVPA as input, for which all the call transitions $(q,a,q',\gamma)$ are such that $q'=q_0$ (see \Cref{def:1SEVPA}). Therefore some simplifications can be done in the algorithm, namely:
\begin{itemize}
    \item in Case (1), $\setOfStates$ is updated to $\Rupdate = \identity{\{q_0\}}$,
    \item in Case (3), $\setOfStates$ is updated to $\Rupdate = \{(p,q) \mid \exists (p,p') \in \setOfStates', (p',\lcrochet,q_0,\gamma) \in \callFunction, (q_0,r) \in \setOfStates, (r,\rcrochet,\gamma,q) \in \returnFunction \}$,
    \item in Case (7), if $(\setOfStates',\laccol)$ is popped, then
        \[
            \Rupdate = \{(p,q) \mid \exists (p,p') \in \setOfStates', (p',\laccol,q_0,\gamma) \in \callFunction, (q_0,\raccol,\gamma,q) \in \returnFunction \},
        \]
        otherwise $\setOfStates$ is updated to
        \begin{align*}
            \Rupdate = \{
                (p,q) \mid
                    &\exists (p,p') \in \setOfStates',
                    (p',\laccol,q_0,\gamma) \in \callFunction,\\
                    &(q_0,r) \in \validPaths(K,\Marks),
                    (r,\raccol,\gamma,q) \in \returnFunction
            \}.
        \end{align*}
        In particular, the computation of $\validPaths(K,\Marks)$ can be restricted to pairs $(r,r')$ such that $r = q_0$.
\end{itemize}

The resulting algorithm is given in \Cref{alg:validating}. This algorithm does not include the \emph{preprocessing algorithm} of learning a 1-SEVPA $\automaton$ from the JSON schema $\schema$ and of computing the key graph $\keyGraph$. 

\begin{algorithm} 
    \caption{Streaming algorithm for validating JSON documents against a given JSON schema.}%
    \label{alg:validating}
    \begin{algorithmic}[1]
        \Require A 1-SEVPA $\automaton$ over $\JSONpushdownAlphabet$ accepting $\languageOrderedOf{\grammar}$ for a closed extended CFG $\grammar$ defining a JSON schema, its key graph $\keyGraph$, and a word $w \in \JSONalphabet^* \setminus \{\emptyword\}$.
        \Ensure A Boolean that is true if and only if $w \in \languageOf{\grammar}$.
        \Statex
%        \State Let $\keyGraph$ be the pre-computed key graph for $\automaton$
        \State $\langle \setOfStates, \St \rangle \gets \langle \identity{\initialState}, \emptyword \rangle, ~~~~a \gets w_1$ \Comment{$w_i$ is the $i$th symbol of $w$}
%        \State $a \gets w_1$ 
        \ForAll{$i \in \{2, \ldots, \lengthOf{w}\}$}
            \IfThenElse{$i \leq \lengthOf{w}$}{$b \gets w_i$}{$b \gets \emptyword$}
            \If{$a \in \callAlphabet$}\label{alg:validation:call}
                \If {$a = \lcrochet$}\label{alg:validation:call:crochet}
                    \State $\Call{Push}{\St, (\setOfStates, \lcrochet)}, ~~~~\setOfStates \gets \identity{\{\initialState\}}$
        %            \State $\setOfStates \gets \identity{\{\initialState\}}$
                \Else\label{alg:validation:call:accol}  \Comment{As $a \in \callAlphabet$, we have $a = \laccol$}
                    \If{$b = \raccol$} 
                        \State $\Call{Push}{\St, (\setOfStates, \laccol)}, ~~~~\setOfStates \gets \identity{\{\initialState\}}$
      %                  \State $\setOfStates \gets \identity{Q}$
                    \ElsIf{$b \in \keyAlphabet$}
                        \State $\Call{Push}{\St, (\setOfStates, \laccol, \{b\}, b, \emptyset)}, ~~~~\setOfStates \gets \identity{P_b}$
       %                 \State $\setOfStates \gets \identity{P_k}$
                    \Else \SpaceReturn false
                    \EndIf
                \EndIf
            \ElsIf{$a \in \returnAlphabet$}\label{alg:validation:return}
                \If {($a = \rcrochet$ (resp. $\raccol$) and $\lcrochet$ (resp. $\laccol$) does not appear on top of $\St$)}
                    \Return false
                    \EndIf
                \If {($a = \rcrochet$)}\label{alg:validation:return:crochet}
                    \State $(\setOfStates', \lcrochet) \gets \Call{Pop}{\St}$ 
                    \State $\setOfStates \gets \{(p, q) \mid (p, p') \in \setOfStates', (p', \lcrochet, q_0, \gamma) \in \callFunction, (q_0, r) \in \setOfStates, (r, \rcrochet, \gamma, q) \in \returnFunction\}$
                \ElsIf{($a = \raccol$ and $(\setOfStates',\laccol)$ is on top of $\St$ for some $\setOfStates'$)}\label{alg:validation:return:accol1}
                    \State $(\setOfStates', \laccol) \gets \Call{Pop}{\St}$
                    \State $\setOfStates \gets \{(p, q) \mid (p, p') \in \setOfStates', (p', \laccol, q_0, \gamma) \in \callFunction, (q_0, \raccol, \gamma, q) \in \returnFunction\}$
                \Else\label{alg:validation:return:accol2}
                    \State $(\setOfStates', \laccol, K, k, \Marks) \gets \Call{Pop}{\St}$ 
                    \State $\Marks \gets \Marks \cup \{(p, k, p') \in \keyGraph \mid (p, p') \notin \setOfStates\}$\label{alg:validation:return:accol:marks}
                    \State $V \gets \validPaths(K,\Marks)$\label{alg:validation:validpaths}
                    \State $\setOfStates \gets \{
                        (p, q) \mid
                            (p, p') \in \setOfStates',
                            (p', \laccol, q_0, \gamma) \in \callFunction,
                            (q_0, r) \in V,
                            (r, \raccol, \gamma, q) \in \returnFunction
                    \}$
                \EndIf
            \Else \Comment{As $a \not\in \callAlphabet \cup \returnAlphabet$, we have $a \in \internalAlphabet$}
                \If {($\lcrochet$ appears on top of $\St$) or ($a \neq \comma$ and $\laccol$ appears on top of $\St$)}\label{alg:validation:intern}
                     \State $\setOfStates \gets \{(p, q) \mid (p, p') \in \setOfStates, (p', a, q) \in \internalFunction\}$
                \ElsIf{($a = \comma$ and $\laccol$ appears on top of $\St$)}\label{alg:validation:comma} 
                    \State $(\setOfStates', \laccol, K, k, \Marks) \gets \Call{Pop}{\St}$
                    \If {($b \notin \keyAlphabet$ or $b \in K$)}
                    \Return false
                    \EndIf
                    \State $K \gets K \cup \{b\}$
                    \State $\Marks \gets \Marks \cup \{(p, k, p') \in \keyGraph \mid (p, p') \notin \setOfStates\}$\label{alg:validation:comma:marks}
                    \State $\Call{Push}{\St, (\setOfStates', \laccol, K, b, \Marks)}, ~~~~\setOfStates \gets \identity{P_b}$
                \Else \SpaceReturn false
            \EndIf
            \EndIf
            \State $a \gets b$
        \EndFor
            \IfThenElse{($\St = \emptyword$ and $\exists (q_0,q) \in \setOfStates$ with $q \in \finalStates$)}{\Return true}{\Return false}
    \end{algorithmic}
\end{algorithm}

\subsection{Complexity of the Algorithm}\label{subsec:Complexity}
%=======================================

In this section, we study the time complexity of our validation algorithm as well as the size of the used memory. %We begin with the time complexity.

\begin{proposition}\label{prop:timeComplexityValidation}
Let $\schema$ be a JSON schema defined by a closed extended CFG $\grammar$ and $\automaton$ be a 1-SEVPA $\automaton$ accepting $\languageOrderedOf{\grammar}$. Checking whether a JSON document $J$ satisfies the schema $\schema$ is in time $\complexity(|J|\cdot (|\states|^4 + |\states|^{|\keyAlphabet|} \cdot |\keyAlphabet|^{|\keyAlphabet|+1}))$.
\end{proposition}

\begin{proof}
Before studying the complexity of \Cref{alg:validating}, let us mention that in addition to the key graph $\keyGraph$, for each key $k \in \keyAlphabet$, we have a list, denoted by $\List{k}$, in which all the vertices in $\keyGraph$ of the form $(p,k,p')$ are stored. Those lists are useful to compute the set $\Marks$ (see steps~\ref{alg:validation:return:accol:marks} and~\ref{alg:validation:comma:marks}). 

Let us also comment on how the set $\validPaths(K,\Marks)$ is computed in step~\ref{alg:validation:validpaths}. Recall that $\keyGraph$ has a finite number of paths (see \Cref{cor:acyclic}) and that each element $(r,r')$ of $\validPaths(K,\Marks)$ is such that $r = q_0$. By a recursive algorithm, we visit each path of $\keyGraph$ starting with any vertex of the form $(p,k,p')$ with $p = q_0$. We stop visiting such a path as soon as we visit a vertex containing a key $k \not\in K$ or belonging to $\Marks$. During the visit of the current path, we collect the keys appearing in its vertices in a set $K'$. When the path reaches some vertex $(r,k,r')$ with $K' = K$, then we add $(q_0,r')$ to $\validPaths(K,\Marks)$. Hence computing $\validPaths(K,\Marks)$ is in $\complexity(|\keyAlphabet| \cdot \validComplexity) = \complexity(|\states|^{|\keyAlphabet|} \cdot |\keyAlphabet|^{|\keyAlphabet|+1})$ by \Cref{lem:sizeKeyGraph} and because checking equality $K' = K$ in $\complexity(|\keyAlphabet|)$.

Let us now consider each case of \Cref{alg:validating} and study its complexity (when false is not returned). Recall that $\automaton$ is deterministic meaning that given a left-hand side of a transition, we have access in constant time to its right-hand side. Notice that at several places, the current set $\setOfStates \subseteq \states^2$ is updated as $\identity{P}$ for some subset $P \subseteq \states$, that can be done in $\complexity(|\states|^2)$. The different cases are the following ones:
\begin{itemize}
    \item The cases $a = \lcrochet$ (line~\ref{alg:validation:call:crochet}) and $a = \laccol$ (line~\ref{alg:validation:call:accol}) are in $\complexity(|\states|^2)$.
    \item The case $a = \rcrochet$ (line~\ref{alg:validation:return:crochet}) is in $\complexity(|\setOfStates|\cdot|\setOfStates'|)$ for computing the updated set $\Rupdate = \{(p,q) \mid \exists (p,p') \in \setOfStates', (p',\lcrochet,q_0,\gamma) \in \callFunction, (q_0,r) \in \setOfStates, (r,\rcrochet,\gamma,q) \in \returnFunction \}$. Indeed we have access in constant time to $(p', \lcrochet, q_0, \gamma) \in \callFunction$ and $(r, \rcrochet, \gamma, q) \in \returnFunction$ when computing $\Rupdate$. Therefore this case is in $\complexity(|\states|^4)$.
    \item The case $a \in \internalAlphabet$ with $\lcrochet$ appearing on top of $\St$ or $a \neq \comma$ if $\laccol$ appears on top of $\St$ (line~\ref{alg:validation:intern}) is in $\complexity(|\states|^2)$.
    \item The case $a = \comma$ (line~\ref{alg:validation:comma}) is in $\complexity(|\states|^2)$. Indeed finding the vertices $(p,k,p')$ that have to be added to $\Marks$ can be done in $\complexity(|\states|^2)$ by traversing the list $\List{k}$.
    \item The case $a = \raccol$ (lines~\ref{alg:validation:return:accol1} and~\ref{alg:validation:return:accol2}) has some similarities with the case $a = \rcrochet$. In case $(\setOfStates',\laccol)$ is popped, then this step is in $\complexity(|\states|^2)$ as $\setOfStates$ does not appear in the computation of $\Rupdate$. In the other case, we need to compute the set $\validPaths(K,\Marks) \subseteq \states^2$. This step is thus in $\complexity(|\states|^4 + |\states|^{|\keyAlphabet|} \cdot |\keyAlphabet|^{|\keyAlphabet|+1})$.%, thus in $\complexity(\validComplexity)$.
\end{itemize}

Therefore, the overall time complexity of \Cref{alg:validating} is in $\complexity(|w|\cdot (|\states|^4 + |\states|^{|\keyAlphabet|} \cdot |\keyAlphabet|^{|\keyAlphabet|+1}))$.
\qed\end{proof}

We now proceed with the memory complexity.

\begin{proposition}
Let $\schema$ be a JSON schema defined by a closed extended CFG $\grammar$ and $\automaton$ be a 1-SEVPA $\automaton$ accepting $\languageOrderedOf{\grammar}$. Checking whether a JSON document $J$ with \depth\ $\depthSymbol(J)$ satisfies the schema $\schema$ uses an amount of memory in $\complexity(|\transitionFunction| + \automatonSize^2 \cdot |\keyAlphabet| + \depthSymbol(J) \cdot (\automatonSize^2 + |\keyAlphabet|))$.
\end{proposition}

\begin{proof}
\Cref{alg:validating} uses auxiliary memory to store:
\begin{itemize}
    \item the given 1-SEVPA $\automaton$,
    \item its key graph $\keyGraph$ with the lists $\List{k}$, $k \in \keyAlphabet$, as introduced in the previous proof,
    \item the current set $\setOfStates$, the current symbol $a$ and the symbol $b$ following $a$,
    \item the stack $\St$ whose elements are of the form either $(\setOfStates',\lcrochet)$, or $(\setOfStates',\laccol)$, or $(\setOfStates', \laccol, K, k, \Marks)$,
    \item the set $\validPaths(K,\Marks)$.
\end{itemize}
The 1-SEVPA $\automaton$ has $|\states|$ states and $|\transitionFunction|$ transitions. The key graph has $\complexity(\automatonSize^2 \cdot |\keyAlphabet|)$ vertices by \Cref{lem:sizeKeyGraph} (it is not necessary to count its transitions as it is acyclic, see \Cref{cor:acyclic}). We can use the same bound $\complexity(\automatonSize^2 \cdot |\keyAlphabet|)$ for the lists $\List{k}$, $k \in \keyAlphabet$, as they are together composed of the vertices of $\keyGraph$. The sizes of $\setOfStates$ and $\validPaths(K,\Marks)$ are in $\complexity(\automatonSize^2)$ as they are subsets of $Q^2$. The biggest elements in the stack $\St$ are of the form $(\setOfStates', \laccol, K, k, \Marks)$ with $K \subseteq \keyAlphabet$ and $\Marks$ containing some vertices of $\keyGraph$, thus with a size in $\complexity(\automatonSize^2 + |\keyAlphabet|)$. The number of elements stored in the stack $\St$ is bounded by the \depth{} $\depthSymbol(J)$ of the JSON document $J$. All in all, the memory used by \Cref{alg:validating} is in $\complexity(|\transitionFunction| + \automatonSize^2 \cdot |\keyAlphabet| + \depthSymbol(J) \cdot (\automatonSize^2 + |\keyAlphabet|))$.
\qed\end{proof}

\section{Implementation and Experiments}\label{sec:implementation}
%======================================

In this section, we discuss the Java implementation of our framework.
That is, we implemented a way to learn a 1-SEVPA $\automaton$ from a JSON schema, the construction of the key graph $\keyGraph$, and our validation algorithm.
First, in \Cref{sec:implementation:classic_validation}, we explain how to validate a JSON document against a schema with the classical algorithm used in many implementations. This algorithm will be compared with our validation algorithm.
Second, in order to learn an automaton, we must be able to find a counterexample if the hypothesis provided by the learner for an equivalence query is incorrect.
The approach we used to generate JSON documents is defined in \Cref{sec:implementation:generation}, while \Cref{sec:implementation:learning} explains how the membership and equivalence queries are implemented.
Finally, \Cref{sec:implementation:results} gives the experimental results for the learning process, and for the comparison of both classical and new validation algorithms.
%The reader is referred to the code documentation for more details about our implementation.\footnote{Our repositories can be consulted at \url{https://github.com/DocSkellington/automatalib}, \url{https://github.com/DocSkellington/learnlib}, \url{https://github.com/DocSkellington/JSONSchemaTools} and \url{https://github.com/DocSkellington/ValidatingJSONDocumentsWithLearnedVPA}.}
The reader is referred to the code documentation for more details about our implementation~\cite{Staquet_AutomataLib,Staquet_JSON_Schema_Tools,Staquet_LearnLib,Staquet_Validating_JSON_Documents}.

In the remaining of this section, let us assume we have a JSON schema $\schema_0$.

\subsection{Classical Validation Algorithm}% to Validate JSON Documents}
\label{sec:implementation:classic_validation}
%===========================================================

Let us briefly explain the \emph{classical} algorithm used in many implementations for validating a JSON document $J_0$ against a JSON schema $S_0$~\cite{JSONSchemaSite}. It is a recursive algorithm that follows the semantics of a closed extended CFG $\grammar$ defining this schema (see \Cref{subsec:AbstractJSON}). For instance, if the current value $J$ is an object, we iterate over each key-value pair in $J$ and its corresponding sub-schema in the current schema $\schema$.\footnote{The sub-schema corresponding to each key-value pair is obtained thanks to a hash map based on the keys.} Then, $J$ satisfies $\schema$ if and only if the values in the key-value pairs all satisfy their corresponding sub-schema. As long as the grammar $\grammar$ does not contain any Boolean operations, this algorithm is straightforward and linear in the size of both the initial document $J_0$ and schema $S_0$. However, if $\grammar$ contains Boolean operations, then the current value $J$ may be processed multiple times. For instance, to verify whether $J$ satisfies $\schema_1 \land \schema_2 \land \dotsb \land \schema_n$, $J$ must be validated against each $\schema_i$.\footnote{Such a recursive algorithm is briefly presented in~\cite{DBLP:conf/www/PezoaRSUV16}.}

In order to match the abstractions we defined (see \Cref{subsec:AbstractJSON}) and to have options to tune the learning process, we implemented our own classical validator.
This implementation includes some optimizations:
\begin{itemize}
    \item If an object does not contain all the required keys, we immediately return false, without needing to validate each key-value pair.
    \item If an object or an array has too few or too many elements (with regards to the constraints defined in the schema), we also immediately return false.
\end{itemize}

\subsection{Generating JSON Documents}\label{sec:implementation:generation}
%==========================================================================

The learning process of a 1-SEVPA from a JSON schema $S_0$ requires to generate JSON documents that satisfy or do not satisfy $S_0$ (see \Cref{subsec:VPAandJSON}). We briefly explain in this section the generators that we implemented, first for generating valid documents and then for generating invalid documents. 

If the grammar $\grammar$ defining the $\schema_0$ does not contain any Boolean operations, generating a document satisfying the schema is easy by following the semantics of $\grammar$ as explained in the previous section. Let us roughly explain how the generation works when $\grammar$ contains Boolean operations and $S_0$ is \emph{not recursive}\footnote{The recursive schemas will be discussed later.}. We follow the productions of $\grammar$ in a top-down manner as follows. If the current production is a disjunction $\nonTerminal \Coloneqq \nonTerminal_1 \lor \nonTerminal_2 \lor \dotsb \lor \nonTerminal_n$, then we select one $\nonTerminal_i$ and replace this production by $\nonTerminal \Coloneqq \nonTerminal_i$. It may happen that the chosen $\nonTerminal_i$ leads to no valid JSON document. Thus we may need to try multiple $\schema_i$ before successfully generating a valid document. In case of a production $\nonTerminal \Coloneqq \nonTerminal_1 \land \nonTerminal_2 \land \dotsb \land \nonTerminal_n$ or $\neg S_1$, then we propagate these Boolean operations lower in the grammar by replacing each $\nonTerminal_i$ by its production and by rewriting the resulting right-hand side as a disjunction. For instance, if $\nonTerminal \Coloneqq \nonTerminal_1 \land \nonTerminal_2$ with $\nonTerminal_1 \Coloneqq \verb!i!$ and $\nonTerminal_2 \Coloneqq \laccol k \nonTerminal_3 \raccol \vee \verb!s!$, then we get 
$\nonTerminal \Coloneqq (\verb!i! \land \laccol k \nonTerminal_3 \raccol) \lor (\verb!i! \land \verb!s!)$. In this simple example, no choice in the disjunction leads to a valid document. In general, the propagation of Boolean operations may require several iterations before yielding a JSON document.
We implemented two types of generator, supporting all Boolean operations: (1) a \emph{random} generator where each choice in disjunctions is made at random, (2) an \emph{exhaustive} generator that explores every choice, thus producing every valid document one by one.

We also implemented modifications of these generators to allow the creation of invalid documents. The idea is to follow the same algorithm as before but to sometimes \emph{deviate} from the grammar $\grammar$. For instance, if the current production describes an integer, we can decide to instead generate an array or a string. Moreover, among a predefined finite set of possible deviations, we can either choose one randomly or exhaustively explore each of them, in a way similar to mutation testing~\cite{DBLP:journals/computer/DeMilloLS78,DBLP:journals/tse/JiaH11}. As for valid documents, we respectively call random and exhaustive those two generators. We also allow the user to set a \emph{maximal depth} (i.e., the maximal number of nested objects or arrays). When a generator reaches the maximal depth, it can no longer produce objects or arrays. This is useful for recursive schemas or when generating invalid documents, as it permits us to be sure we eventually produce a document of finite depth.

It is noteworthy that the implementation supports an abstracted version of enumerations of JSON values inside documents.
Recall that, for instance, all strings are abstracted by \texttt{s}.
In a similar way, an enumeration is abstracted by \texttt{e}.
While this implies that the exact values are lost, it allows us to keep the enumerations inside the considered schemas.
Moreover, if the schemas contains regular expressions to define the keys inside an object\footnote{Using the \texttt{patternProperties} field.}, the regular expression is directly used as the key.
That is, we do not generate a string that matches the expression but uses the expression itself as the key.

\subsection{Learning Algorithm}\label{sec:implementation:learning}
%=================================================================

Let us now focus on the learning algorithm itself, and in particular on the membership and equivalence queries. We recall that the equivalence queries are performed by generating a certain number of (valid and invalid) JSON documents and by verifying that the learned VPA $\automaton[H]$ and the given schema $\schema_0$ agree on the documents' validity (see \Cref{subsec:VPAandJSON}). As said in \Cref{sec:definitions:learning}, we use the TTT  algorithm~\cite{DBLP:phd/dnb/Isberner15} to learn a 1-SEVPA from $S_0$. More precisely, we rely on the implementation made by Isberner in the well-known Java libraries \LearnLib and \AutomataLib~\cite{10.1007/978-3-319-21690-4_32}. 

Using the fact that the membership and equivalence queries can be defined independently from the actual learning algorithm, we implement the queries as follows. We use the random and exhaustive generators of valid and invalid documents as explained in \Cref{sec:implementation:generation} and we fix two constants $C$ and $D$ depending on the schema to be learned\footnote{The values of $C$ and $D$ are given in Section~\ref{sec:implementation:results}.}. For a \emph{membership} query over a word \(w \in \JSONalphabet^*\), the teacher runs the classical validator described in \Cref{sec:implementation:classic_validation} on $w$ and $\schema_0$. For an \emph{equivalence} query over a learned 1-SEVPA $\hypothesis$, the teacher performs multiple checks, in this order:
\begin{enumerate}
    \item Is there a loop $(q_0,a,q_0)$ over the initial state of $\hypothesis$ reading an internal symbol $a \in \internalAlphabet$?\footnote{We observed such a situation several times from our experimentation.} In that case, we generate a word $w$ accepted by $\hypothesis$.\footnote{This is supported by \LearnLib.} The word $aw$ is then a counterexample since it is also accepted by $\hypothesis$ but it is not a valid document (it is not an object).
    \item Using a (random or exhaustive) generator for valid documents, is there a valid document $w$ that is not accepted by $\hypothesis$? In that case, $w$ is a counterexample. If the used generator is the exhaustive one, we generate every possible valid document one by one over the course of all equivalence queries. We stop using it once all documents have been generated. If the generator is the random one, for each equivalence query, we generate $C$ valid documents for each document depth between 0 and $D$.
    \item Using a (random or exhaustive) generator for invalid documents, is there an invalid document $w$ that is accepted by $\hypothesis$? In that case, $w$ is a counterexample. Both generators are used as explained in the previous step.
    \item In the key graph $\keyGraph[\hypothesis]$, is there a path $((p_1,k_1,p'_1)(p_2,k_2,p'_2)\ldots (p_n,k_n,p'_n))$ with $p_1 = q_0$ such that $k_i = k_j$ for some $i \neq j$ (see \Cref{cor:acyclic})? In that case, from this path, we construct a counterexample, i.e., a word accepted by $\hypothesis$ that is not a valid document. See \Cref{sec:counterexample_from_key_graph} for the details.
\end{enumerate}
If $\hypothesis$ fails every test (i.e., we could not find a counterexample), we conclude that $\hypothesis$ is correct, the equivalence query succeeds, and we finish the learning process.

\subsection{Experimental Evaluation}\label{sec:implementation:results}
%======================================================================

\subsubsection{Evaluated Schemas}
For the experimental evaluation of our algorithms, we consider the following schemas, sorted in increasing size:
\begin{enumerate}
    \item A schema that accepts documents defined recursively.
        Each object contains a string and can contain an array
        whose single element satisfies the whole schema, i.e., this is a recursive list.
        See \Cref{sec:schemas_implementation}.
    \item A schema that accepts documents containing each type of values, i.e., an object, an array, a string, a number, an integer, and a Boolean.
        See \Cref{sec:schemas_implementation}.
    \item A schema that defines how snippets must be described in \emph{Visual Studio Code}.
        This schema can be downloaded from \url{https://raw.githubusercontent.com/Yash-Singh1/vscode-snippets-json-schema/main/schema.json}.
    \item A recursive schema that defines how the metadata files for \emph{VIM plugins} must be written, see \url{https://json.schemastore.org/vim-addon-info.json}.
    \item A schema that defines how \emph{Azure Functions Proxies} files must look like, see \url{https://json.schemastore.org/proxies.json}.
    \item A schema that defines the configuration file for a code coverage tool called \emph{codecov}, see \url{https://json.schemastore.org/codecov.json}.
\end{enumerate}
That is, we consider two schemas written by ourselves to test our framework, and four schemas that are used in real world cases.
The last four schemas were modified to make all object keys mandatory\footnote{That is, we added all the keys in a \texttt{required} field in each object.} and to remove unsupported keywords.
All schemas and scripts used for the benchmarks can be consulted on our repository~\cite{Staquet_Validating_JSON_Documents}.
In the rest of this section, the schemas are referred to by their order in the enumeration.

We present three types of experimental results.
First, we discuss the time and the number of membership and equivalence queries needed to learn a 1-SEVPA %VPA
from a JSON schema. Then, given such a 1-SEVPA \(\automaton\), we give the time and memory required to compute its reachability relation \(\accRelation\) and its key graph \(\keyGraph\).
Finally, the time and memory used to validate a JSON document using the classical and our new algorithms are provided.
The server used for the benchmarks ran Debian 10 over Linux 5.4.73-1-pve with a 4-core Intel\textregistered{} Xeon\textregistered{} Silver 4214R Processor with 16.5M cache, and 64GB of RAM\@. % chktex 8
Moreover, we used OpenJDK version 11.0.12.

\subsubsection{Learning VPAs}
As the first part of the preprocessing step of our validation algorithm, we must learn a 1-SEVPA\@.
For the first three evaluated schemas, we use an exhaustive generator thanks to the small number of valid documents that can be produced (up to a certain depth $D$ defined below).
For the remaining three schemas, we rely instead on a random generator where we fix the number of generated documents per round at \(C = 10000\).
For both exhaustive and random generators, the maximal depth of the generated documents is set at \(D = \depth(S) + 1\), where \(\depth(S)\) is the maximal number of nested defined objects and arrays in the schema \(S\), except for the recursive list schema where $D = 10$, and for the recursive \emph{VIM plugins} schema where $D = 7$. 

In all cases, we learn the 1-SEVPA while measuring the total time (including the time needed for the membership and equivalence queries), and the total number of queries.
After removing the bin state of the resulting 1-SEVPA (see \Cref{app:ComplexityKeyGraph}), we also measure its number of states and transitions, as well as its diameter.
For the first five schemas, we do not set a time limit and repeat the learning process ten times. For the last schema, we set a time limit of one week and, for time constraints, only perform the learning process once.
After that, we stop the computation and retrieve the learned 1-SEVPA at that point.

Recall that for one of the checks\footnote{Check 4., see Section~\ref{sec:implementation:learning}.} performed in an equivalence query, we must construct the key graph of the hypothesis $\hypothesis$.
In particular, we must compute its reachability relation $\accRelation[\hypothesis]$.
As we observe that a counterexample implies only a small change in $\hypothesis$, we gain efficiency by  avoiding to compute the new reachability relation from scratch at each equivalence query.

The results for the learning benchmarks are given in \Cref{tab:benchmarks:learning}. The learning process for the last schema was not completed after one week. The retrieved learned 1-SEVPA is therefore an approximation of this schema. Its key graph has repeated keys along some of its paths, a situation that cannot occur if the 1-SEVPA was correctly learned, see \Cref{cor:acyclic}.
Notice that, when using a random generator, different runs of the algorithm may yield different 1-SEVPAs.

\begin{table}
    \centering
    \begin{tabular}{rrrrrrrrrrr}
    \toprule
    Time (s) & Membership  & Equivalence & $|\states|$ & $|\Sigma|$ & $|\delta_{c}|$ & $|\delta_{r}|$ & $|\delta_i|$ & Diameter \\
    \midrule
    2.2      & 2055.0      & 5.0         & 7.0            & 15.0       & 14.0           & 3.0            & 5.0          & 3.0      \\
    4.5      & 69514.0     & 3.0         & 24.0           & 20.0       & 48.0           & 3.0            & 26.0         & 12.0     \\
    9.0      & 21943.0     & 5.0         & 16.0           & 17.0       & 32.0           & 7.0            & 18.0         & 13.0     \\
    9590.3   & 4246085.0   & 36.4        & 150.0          & 27.0       & 300.0          & 2946.5         & 760.3        & 9.0      \\
    35008.2  & 4063971.7   & 30.5        & 121.0          & 35.0       & 242.0          & 2123.0         & 752.5        & 13.3     \\
    Timeout  & 633049534.0 & 192.0       & 884.0          & 77.0       & 1768.0         & 89695.0        & 8557.0       & 28.0     \\
    \bottomrule
\end{tabular}
    \caption{Results of the learning benchmarks. For the first five schemas, values are averaged out of ten experiments. For the last schema, only one experiment was conducted.}%
    \label{tab:benchmarks:learning}
\end{table}

\subsubsection{Construction of the Key Graph}
The second part of the preprocessing step is to construct the key graph of the learned 1-SEVPA\@.
For each evaluated schema, we select the learned 1-SEVPA with the largest set of states, in order to report a worst-case measure.
From that 1-SEVPA $\automaton$, we compute its reachability relation $\accRelation$ and its key graph $\keyGraph$, and measure the time and memory used, as well as their sizes. Values obtained after a single experiment are given in \Cref{tab:benchmarks:preprocessing}. We can see that the size of the key graph is far from the theoretical upper bound in $\complexity(|\states|^2 \cdot |\keyAlphabet|)$ (see Lemma~\ref{lem:sizeKeyGraph}), and that its storage does not consume more than one megabyte, except for \emph{codecov} schema. That is, even for non-trivial schemas, the key graph is relatively lightweight.

\begin{table}
    \centering
    \begin{tabular}{rr@{\hspace{5pt}}r @{\hspace{20pt}} rrr@{\hspace{5pt}}r}
    \toprule
    \multicolumn{3}{c}{$\accRelation$} & \multicolumn{4}{c}{$\keyGraph$}                                                                        \\
    \cmidrule(r{20pt}){1-3} \cmidrule{4-7}
    Time (s)                           & Memory (kB)                     & Size   & Time (s) & Computation (kB)                     & Storage (kB) & Size \\
    \midrule
    34                                 & 492                             & 31     & 100      & 2231                            & 65      & 3    \\
    67                                 & 1152                            & 213    & 234      & 2623                            & 69      & 9    \\
    67                                 & 737                             & 125    & 118      & 2223                            & 69      & 10   \\
    1756                               & 10316                           & 5832   & 1715     & 11827                           & 419     & 418  \\
    2208                               & 13978                           & 4420   & 2839     & 17968                           & 667     & 541  \\
    377141                             & 212970                          & 270886 & 187659   & 120398                          & 16335   & 6397 \\
    \bottomrule
\end{tabular}

    \caption{Time and memory needed to compute $\accRelation$ and $\keyGraph$, and their size. Values are taken from a single experiment.
    For $\keyGraph$, the Computation (resp.\ Storage) column gives the memory required to compute $\keyGraph$ (resp.\ to store $\keyGraph$).
    }%
    \label{tab:benchmarks:preprocessing}
\end{table}

\subsubsection{Comparing Validation Algorithms}
Finally, we compare both validation algorithms: the classical one and our new streaming algorithm.
For the latter, we use the 1-SEVPA (and its key graph) selected in the previous section. 
We first generate 5000 valid and 5000 invalid JSON documents using a random generator, with a maximal depth equal to $D = 20$. We then execute both validation algorithms on these documents, while measuring the time and memory required.
%Note that they always agree, i.e., they both accept or reject at the same time.
On all considered documents, both algorithms return the same classification output, even for the partially learned 1-SEVPA.\@

Since obtaining a close approximation of the consumed memory requires Java to \enquote{stop the world} to destroy all unused objects, we execute each algorithm twice: one time where we only measure time, and a second time where we request the memory to be cleaned each iteration of the algorithm (or each recursive call).

For our algorithm, we only measure the memory required to execute the algorithm, as we do not need to store the whole document to be able to process it. We also do not count the memory to store the 1-SEVPA and its key graph.
As the classical algorithm must have the complete document stored in memory, in this case, we sum the RAM consumption for the document and for the algorithm itself. This is coherent to what happens in actual web-service handling: Whenever a new validation request is received, we would spawn a new subprocess that handles a specific document. Since the 1-SEVPA and its key graph are the same for all subprocesses, they would be loaded in a memory space shared by all processes. 

In order to focus on the three largest schemas, we defer the presentation of the first three schemas to \Cref{app:more_results}. Results for \emph{VIM plugins}, \emph{Azure Functions Proxies}, and \emph{codecov} are given in \Crefrange{fig:benchmarks:validating:vim}{fig:benchmarks:validating:codecov}.
The blue crosses give the values for our algorithm, while the red circles stand for the classical algorithm.
The x-axis gives the number of symbols of our alphabet in each document.
We can see that our algorithm usually requires more time than the classical algorithm to validate a document.
We observed that a majority of the total time is spent computing the set $\validPaths(K, \Marks)$.
For both \emph{VIM plugins} and \emph{Azure Functions Proxies}, our algorithm consumes less memory than the classical algorithm. We observed that, even if we add the number of bytes needed to store the automaton and the key graph (see Tables~\ref{tab:benchmarks:learning} and~\ref{tab:benchmarks:preprocessing}), we remain under the classical algorithm, especially for large documents.

For the last \emph{codecov} schema, we recall that the learning process was not completed, leading to an approximated 1-SEVPA whose key graph has repeated keys. We had to modify our validation algorithm, in particular the computation of the set $\validPaths(K,\Marks)$, see proof of \Cref{prop:timeComplexityValidation}: when visiting each path of the key graph, we pay attention to stop it on the second visit of a key. This means that we have to explore some invalid paths, increasing the memory and time consumed by our algorithm. Thus, it appears that, while a 1-SEVPA that is not completely learned can still be used in our algorithm to correctly decide whether a document is valid, stopping the learning process early may drastically increase the time and space required.

\begin{figure}
    \centering
    \begin{subfigure}{0.45\textwidth}
        \includegraphics{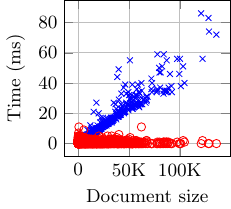}
        \caption{Time.}
    \end{subfigure}
    \hfill
    \begin{subfigure}{0.45\textwidth}
        \includegraphics{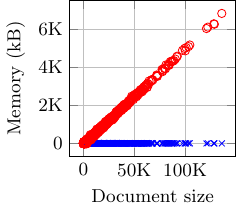}
        \caption{Memory.}
    \end{subfigure}
    \caption{Results of validation benchmarks for the metadata files for \emph{VIM plugins}.%
    %Blue crosses give the values for our algorithm, and red circles the values for the classical validator. Values are averaged out of ten experiments.
    }%
    \label{fig:benchmarks:validating:vim}

    \begin{subfigure}{0.45\textwidth}
        \includegraphics{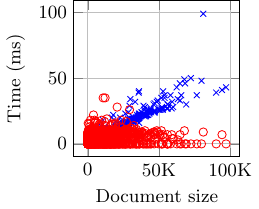}
        \caption{Time.}
    \end{subfigure}
    \hfill
    \begin{subfigure}{0.45\textwidth}
        \includegraphics{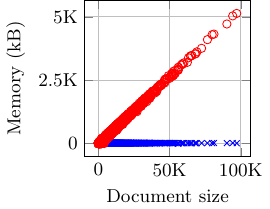}
        \caption{Memory.}
    \end{subfigure}
    \caption{Results of validation benchmarks for the \emph{Azure Functions Proxies} file.%
    %Blue crosses give the values for our algorithm, and red circles the values for the classical validator. Values are averaged out of ten experiments.
    }%
    \label{fig:benchmarks:validating:proxies}

    \begin{subfigure}{0.45\textwidth}
        \includegraphics{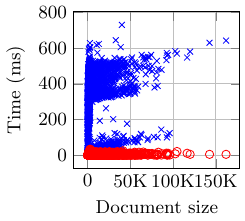}
        \caption{Time.}
    \end{subfigure}
    \hfill
    \begin{subfigure}{0.45\textwidth}
        \includegraphics{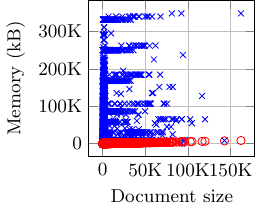}
        \caption{Memory.}
    \end{subfigure}
    \caption{Results of validation benchmarks for the configuration files of \emph{codecov}.%
    %Blue crosses give the values for our algorithm, and red circles the values for the classical validator. Values are averaged out of ten experiments.
    }%
    \label{fig:benchmarks:validating:codecov}
\end{figure}

\subsection{Worst Case for the Classical Validator}

Finally, let us consider one last schema, purposefully made to highlight the differences between both algorithms.
Its extended CFG (up to key permutations) is as follows, with \(\ell \geq 1\):
\[
    \begin{array}{lcll}
        \nonTerminal &\Coloneqq& \nonTerminal_1 \land \nonTerminal_2 \land \dotsb \land \nonTerminal_{\ell}&\\
        \nonTerminal_i &\Coloneqq& R_i \lor R_{i+1} \lor \dotsb \lor R_{\ell} &\forall i \in \{1, \dotsc, \ell\}\\
        R_i &\Coloneqq& \laccol k_i \verb!s! \comma k_{i+1} \verb!s! \comma \dotso \comma k_{\ell} \verb!s! \raccol &\forall i \in \{1, \dotsc, \ell\}
    \end{array}
\]
Notice the difference between $\nonTerminal_i$ and $\nonTerminal_{i+1}$: $R_i$ is removed from $\nonTerminal_i$ to get $\nonTerminal_{i+1}$. Hence an equivalent grammar is $\nonTerminal \Coloneqq R_\ell$. On the one hand, the classical validator has to explore each $\nonTerminal_i$ in the conjunction, and each $R_j$ in the related disjunction, before finally considering $\nonTerminal_\ell$ and thus $R_{\ell}$. On the other hand, as we learn minimal 1-SEVPAs, we can here easily construct a 1-SEVPA directly for the simplified grammar \(\nonTerminal \Coloneqq R_{\ell}\). This example is a worst-case for the classical validator, while being easy to handle for our algorithm, as shown in \Cref{tab:benchmarks:preprocessing:worstcase,fig:benchmarks:validating:worstcase}.

\begin{table}
    \centering
    \begin{tabular}{rr@{\hspace{5pt}}r @{\hspace{20pt}} rrr@{\hspace{5pt}}r}
    \toprule
    \multicolumn{3}{c}{$\accRelation$} & \multicolumn{4}{c}{$\keyGraph$}                                                                        \\
    \cmidrule(r{20pt}){1-3} \cmidrule{4-7}
    Time (s)                           & Memory (kB)                     & Size   & Time (s) & Computation (kB)                     & Storage (kB) & Size \\
    \midrule
    33.0 & 492.0 & 31.0 & 109.0 & 1981.0 & 60.0 & 2.0 \\
    \bottomrule
\end{tabular}

    \caption{Time and memory needed to compute $\accRelation$ and $\keyGraph$, and their size. Values are taken from a single experiment.
    For $\keyGraph$, the Computation (resp. Storage) column gives the memory required to compute $\keyGraph$ (resp.\ to store $\keyGraph$).
    }%
    \label{tab:benchmarks:preprocessing:worstcase}
\end{table}

\begin{figure}
    \begin{subfigure}{0.45\textwidth}
        \includegraphics{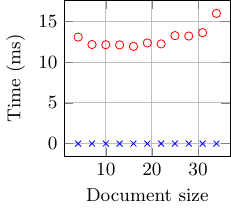}
        \caption{Time.}
    \end{subfigure}
    \hfill
    \begin{subfigure}{0.45\textwidth}
        \includegraphics{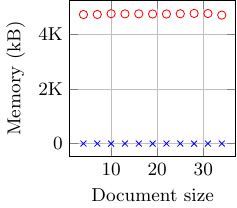}
        \caption{Memory.}
    \end{subfigure}
    \caption{Results of validation benchmarks for the \emph{worst-case} schema, with \(\ell = 10\).%
    %Blue crosses give the values for our algorithm, and red circles the values for the classical validator. Values are averaged out of ten experiments.
    }%
    \label{fig:benchmarks:validating:worstcase}
\end{figure}

\section{Conclusion}
%==================    
In this paper, we have proved that, given any JSON schema, one can construct a VPA that accepts the same set of JSON documents as the schema. Leveraging this fact and \TTT, we designed a learning algorithm that yields a VPA for the schema, under a fixed order of the keys inside objects. We then abstracted as a key graph the part of this VPA dealing with objects, and proposed a streaming algorithm that uses both the graph and the VPA to decide whether a document is valid for the schema, under any order on the keys.

As future work, one could focus on constructing the VPA directly from the schema, without going through a learning algorithm. While this task is easy if the schema does not contain Boolean operations, it is not yet clear how to proceed in the general case. Second, it could be worthwhile to compare our algorithm against an implementation of a classical algorithm used in the industry. This would require either to modify the industrial implementations to support abstractions, or to modify our algorithm to work on unabstracted JSON schemas. Third, in our validation approach, we decided to use a VPA accepting the JSON documents satisfying a fixed key order --- thus requiring to use the key graph and its costly computation of the set $\validPaths(K,\Marks)$. It could be interesting to make additional experiments to compare this approach with one where we instead use a VPA accepting the JSON documents and all their key permutations --- in this case, reasoning on the key graph would no longer be needed.
%The key graph would then be no longer necessary but the learning process of the VPA could be much longer. 
Finally, motivated by obtaining efficient querying algorithms on XML trees, the authors of~\cite{DBLP:conf/birthday/SeidlSM08} have introduced the concept of mixed automata in a way to accept subsets of unranked trees where some nodes have ordered sons and some other have unordered sons. It would be interesting to adapt our validation algorithm to different formalisms of documents, such as the one of mixed automata.
 
%% Bibliography
\bibliographystyle{splncs04}
\bibliography{references}

\appendix
    
\section{Interest of Fixing a Key Order}\label{app:expSmaller}
%==========================================
In this section, we provide a family of JSON schemas $S_n$, $n \geq 1$, and a key-order, such that the minimal 1-SEVPA $\automaton_n$ accepting all the JSON documents satisfying $S_n$ is exponentially larger than the minimal 1-SEVPA ${\automaton[B]}_n$ accepting those documents respecting the key order.

Let $\keyAlphabet = \{k_1,\ldots,k_n\}$ with the key order $k_1 < \dotsb < k_n$. The proposed JSON schema $S_n$ is the one defining all JSON documents $\laccol k_{i_1} \verb!s! \comma \ldots \comma k_{i_n} \verb!s! \raccol$ where $(k_{i_1}, \ldots, k_{i_n})$ is a permutation of $(k_1,\ldots,k_n)$. It is easy to see that the 1-SEVPA ${\automaton[B]}_n$ has $\complexity(n)$ states since it only accepts the document $\laccol k_{1} \verb!s! \comma \ldots \comma k_{n} \verb!s! \raccol$. Let us show that the 1-SEVPA ${\automaton[A]}_n$ has at least $2^n$ states.

Let us consider the state $q$ of $\automaton_n$ such that
    \[
        \langle q_0, \emptyword \rangle \xrightarrow{\laccol} \langle q_0, \gamma \rangle \xrightarrow{x} \langle q, \gamma \rangle \xrightarrow{y} \langle p, \gamma \rangle \xrightarrow{\raccol} \langle r, \emptyword \rangle
    \]
with $r \in \finalStates$, $x = k_{i_1} \verb!s! \comma \ldots \comma k_{i_\ell} \verb!s!$, and $y = k_{i_{\ell + 1}} \verb!s! \comma \ldots \comma k_{i_n} \verb!s!$ where $(k_{i_1}, \ldots, k_{i_\ell}, \ldots, k_{i_n})$ is a permutation of $(k_1,\ldots,k_n)$. It is not possible to have another path 
    \[
        \langle q_0, \gamma \rangle \xrightarrow{x'} \langle q, \gamma \rangle
    \]
    with $x' = k_{j_1} \verb!s! \comma \ldots \comma k_{j_m} \verb!s!$ such that $\{k_{i_1}, \ldots, k_{i_\ell} \} \neq \{k_{j_1}, \ldots, k_{j_m}\}$. 
Otherwise, ${\automaton[A]}_n$ would accept an invalid document. It follows that there are as many such states $q$ as there are subsets $\{k_{i_1}, \ldots, k_{i_\ell} \}$ of $\keyAlphabet$. Therefore ${\automaton[A]}_n$ has at least $2^n$ states.

\section{Complexity of the Key Graph}\label{app:ComplexityKeyGraph}
%============================================================================

We recall that the algorithm for computing the key graph $\keyGraph$ from the 1-SEVPA $\automaton$ requires three main steps: (1) compute the relation $\accRelation$,  (2) detect the bin state and remove it from $\automaton$ (it is unique, if it exists), (3) compute the vertices and edges of $\keyGraph$. We proceed as follows and give the related time complexities.\footnote{Recall that $\automaton$ is deterministic meaning that given a left-hand side of a transition, we have access in constant time to its right-hand side.}

\subsection{Relation $\accRelation$}
First, we enrich the reachability relation $\accRelation$ with words as follows. We define a map $\wit : \accRelation \rightarrow \alphabet^*$ such that for all $(q,q') \in \accRelation$, we have $\wit(q,q') = w$ with $\left\langle q, \emptyword \right\rangle \xrightarrow{w} \left\langle q', \emptyword \right\rangle$. The word $w$ is a \emph{witness} of the membership of $(q,q')$ to $\accRelation$. 
We compute this map as follows by choosing the witnesses step by step while computing $\accRelation$ (an algorithm computing $\accRelation$ is given in Section~\ref{subsec:KeyGraphComputation}).
\begin{itemize}
    \item Initially, for all $q \in Q$, we define $\wit(q,q) = \emptyword$, and for each $(q,q')$ such that $q \neq q'$ and there exists $(q,a,q') \in \internalFunction$, we define $\wit(q,q') = a$.
    \item During the computation of $\accRelation$, a new element $(q,q')$ is added to $\accRelation$ 
    \begin{enumerate}
        \item when there exist $(q,a,p,\gamma) \in \callFunction$, $(p,p') \in \accRelation$, and $(p',\bar a,\gamma,q') \in \returnFunction$,
        \item or when there exist $(q,p), (p,q') \in \accRelation$ (during the transitive closure).
    \end{enumerate}
    In the first case, if $\wit(p,p') = w$, then we define $\wit(q,q') = a w \bar a$. In the second case, if $\wit(q,p) = w$ and $\wit(p,q') = w'$, then we define $\wit(q,q') = ww'$.
\end{itemize}

\begin{lemma}\label{lem:ComputingReach}
Computing the relation $\accRelation$ enriched with the map $\wit$ is in time $\complexity(|\states|^5 + |\transitionFunction| \cdot |\states|^4)$.
\end{lemma}

\begin{proof}
The relation $\accRelation$ with its witnesses are stored in a matrix of size $|\states|^2$. The initialization is in $\complexity(|\states|^2 + |\internalFunction|)$. The main loop uses at most $|\states|^2$ steps. One operation inside this loop is the transitive closure that can be computed in $\complexity(|\states|^3)$ with Warshall's algorithm. The other operation is in $\complexity(|\callFunction| \cdot |\states|^2)$. Therefore, the overall complexity of computing $\accRelation$ is in $\complexity(|\states|^5 + |\transitionFunction| \cdot |\states|^4)$.
\qed\end{proof}

\subsection{Bin State}
Second, for detecting the bin state of $\automaton$, we define the following set $\setR \subseteq Q$: 
    \[
        \setR = \{p \in \states \mid \exists \left\langle q_0, \emptyword \right\rangle \xrightarrow{w} \left\langle q_0, \stack \right\rangle \mbox{ and } \left\langle p, \stack \right\rangle \xrightarrow{w'} \left\langle q, \emptyword \right\rangle  \mbox{ with } q \in \finalStates\},
    \]
enriched with the witness map $\wit' : \setR \rightarrow {(\alphabet^*)}^2$ that assigns to each $p \in \setR$ a pair of words $(w,w')$ as in the previous definition. In the definition of $\setR$, notice the same stack content $\stack$ in both configurations $\left\langle q_0, \stack \right\rangle$ and $\left\langle p, \stack \right\rangle$, and the presence of $q_0$ in the first configuration\footnote{We can restrict to $q_0$, as we work with 1-SEVPAs having all their call transitions going to $q_0$.}.

From this set $\setR$, we easily derive the next lemma. 

\begin{lemma}\label{lem:BinState}
For each state $p \in \states$, $p$ is not a bin state if and only if $p \in \setR$.
\end{lemma}
\begin{proof}
Assume $p$ is not a bin state, i.e., there exists a path $\left\langle q_0, \emptyword \right\rangle \xrightarrow{t} \left\langle p, \stack \right\rangle \xrightarrow{t'} \left\langle q, \emptyword \right\rangle$ in $\transitionSystem$, with $q \in \finalStates$. By repeating the argument given in the proof of Lemma~\ref{lem:abstractKeyGraph}, see~\eqref{eq:noBinState}, we get that this path decomposes as $\left\langle q_0, \emptyword \right\rangle \xrightarrow{t''} \left\langle q_0, \stack \right\rangle \xrightarrow{t_{m+1}} \left\langle p, \stack \right\rangle \xrightarrow{t'} \left\langle q, \emptyword \right\rangle$ with $t = t''t_{m+1}$. This shows that $p \in \setR$.

Now, assume $p \in \setR$, i.e., there exists two paths $\langle q_0, \emptyword \rangle \xrightarrow{w} \langle q_0, \sigma \rangle$ and $\langle p, \sigma \rangle \xrightarrow{w'} \langle q, \emptyword \rangle$ with $q \in \finalStates$.
If we prove that $(q_0,p) \in \accRelation$, we are done. As $\automaton$ is minimal, there exists some path $\left\langle q_0, \emptyword \right\rangle \xrightarrow{t} \left\langle p, \stack \right\rangle$ in $\transitionSystem$. As done above with~\eqref{eq:noBinState}, we get that this path decomposes as $\left\langle q_0, \emptyword \right\rangle \xrightarrow{t''} \left\langle q_0, \stack \right\rangle \xrightarrow{t_{m+1}} \left\langle p, \stack \right\rangle$. It follows that $(q_0,p) \in \accRelation$.
\qed\end{proof}

We compute $\setR$ and $\wit'$ as follows:
\begin{itemize}
    \item Initially, add $q$ to $\setR$ for all $q \in \finalStates$, and assign the witness $\wit'(q) = (\emptyword,\emptyword)$ to $q$. 
    \item Then, repeat until $\setR$ stabilizes:
    \begin{enumerate}
    \item If we have $p \in \setR$ with $\wit'(p) = (w,w')$ and $(p',p) \in \accRelation$ with $\wit(p',p) = t'$, then add\footnote{If $p$ already belongs to $\setR$, we do nothing.} $p'$ to $\setR$ with $\wit'(p') = (w,t'w')$.
    \item If we have $p \in \setR$ with $\wit'(p) = (w,w')$, $(q_0,r) \in \accRelation$ with $\wit(q_0,r) = t$, $(r,a,\gamma,q_0) \in \callFunction$ and $(p', \bar a, \gamma, p) \in \returnFunction$, then add $p'$ to $\setR$ with $\wit'(p') = (wta,\bar a w')$.
    \end{enumerate}
\end{itemize}
This algorithm is correct as shown by the next lemma.

\begin{lemma}
    Let \(\setRalgo\) computed by the algorithm above.
    Then, \(\setRalgo = \setR\).
\end{lemma}
\begin{proof}
    It is easy to see that \(\setRalgo \subseteq \setR\).
    Let us prove that \(\setR \subseteq \setRalgo\).
    Let \(p \in \setR\).
    That is, we have \(w, w' \in \Sigma^*, \sigma \in \stackAlphabet^*\) and \(q \in \finalStates\) such that
    \begin{equation}
        \langle q_0, \emptyword \rangle \xrightarrow{w} \langle q_0, \sigma \rangle \text{ and } \langle p, \sigma \rangle \xrightarrow{w'} \langle q, \emptyword \rangle. \label{eq:proof_setR:paths}
    \end{equation}
    We prove by induction over \(\balance(w) = -\balance(w')\) that \(p \in \setRalgo\).
    
    \emph{Base case:} \(\balance(w) = 0\), i.e., \(\sigma = \emptyword\).
        Thus, we can focus on the witnesses \((\emptyword, w')\) instead of \((w, w')\), as \(\langle q_0, \emptyword \rangle \xrightarrow{\emptyword} \langle q_0, \emptyword \rangle\).
        Moreover, we have \(w' \in \wellMatched\) and \((p, q) \in \accRelation\).
        By the initialization, \(q \in \setRalgo\), and by instruction 1., \(p \in \setRalgo\).

    \emph{Induction step:} \(\balance(w) > 0\).
        In that case, we can decompose \(w = t w_1 a w_2\), with \(t \in {(\wellMatched \cdot \callAlphabet)}^*, a \in \Sigma_c, w_1, w_2 \in \wellMatched\), and \(w' = w'' \overline{a} t'\), with \(w'' \in \wellMatched, \overline{a} \in \returnAlphabet\).
        We can thus decompose the paths~\eqref{eq:proof_setR:paths} into
        \begin{gather*}
            \langle q_0, \emptyword \rangle \xrightarrow{t} \langle q_0, \sigma_1 \rangle \xrightarrow{w_1} \langle r, \sigma_1 \rangle \xrightarrow{a} \langle q_0, \sigma \rangle \xrightarrow{w_2} \langle q_0, \sigma \rangle
        \shortintertext{and}
            \langle p, \sigma \rangle \xrightarrow{w''} \langle s, \sigma \rangle \xrightarrow{\overline{a}} \langle s', \sigma_1 \rangle \xrightarrow{t'} \langle q, \emptyword \rangle.
        \end{gather*}
        Therefore, we can focus on the witnesses \((t w_1 a, w')\) instead of \((w, w')\).

        Since \(\langle q_0, \emptyword \rangle \xrightarrow{t} \langle q_0, \sigma_1 \rangle\) and \(\langle s', \sigma_1 \rangle \xrightarrow{t'} \langle q, \emptyword \rangle\), it holds that \(s' \in \setR\) and, by induction, \(s' \in \setRalgo\).
        By instruction 2.\ of the algorithm, \(s \in \setRalgo\).
        Finally, by instruction 1., \(p \in \setRalgo\).

    In conclusion, we have \(p \in \setR \implies p \in \setRalgo\).
\qed\end{proof}

\begin{lemma}\label{lem:ComputingBinState}
Detecting and removing the bin state of $\automaton$, if it exists, is in time $\complexity(|\transitionFunction|\cdot|\states|^4)$.
\end{lemma}

\begin{proof}
Let us first show that computing the set $\setR$ is in $\complexity(|\transitionFunction|\cdot|\states|^4)$. The initialization is in $\complexity(|\states|)$. The main loop uses at most $|\states|$ steps until $\setR$ stabilizes. Its body is in $\complexity(|\states|^3 + |\returnFunction|\cdot|\states|^3)$. Second, detecting the bin state of $\automaton$ is in $\complexity(|\states|^2)$ by Lemma~\ref{lem:BinState}. Finally, removing this bin state from $\automaton$ and its related transitions is in $\complexity(|\states| + |\transitionFunction|)$. Therefore we get the complexity announced in the lemma.
\qed\end{proof}

\subsection{Key Graph}
Third, the vertices and the edges $\keyGraph$ are computed as follows: $(p,k,p')$ is a vertex in $\keyGraph$ if there exist $(p,k,q) \in \internalFunction$ with $k \in \keyAlphabet$ and 
\begin{itemize}
    \item $(q,a,p') \in \internalFunction$  with $a \in \valAlphabet$,
    \item or $(q,a,r,\gamma) \in \callFunction$, $(r,r') \in \accRelation$, and $(r',\bar a,\gamma,p') \in \returnFunction$;
\end{itemize}
and $((p_1,k_1,p'_1),(p_2,k_2,p'_2))$ is an edge in $\keyGraph$ if there exists $(p'_1,\comma,p_2) \in \internalFunction$.

\begin{lemma}\label{lem:ComputingKeyGraph}
Constructing the key graph $\keyGraph$ is in $\complexity(|\transitionFunction|^2 + |\transitionFunction| \cdot |\states|^2 + |\states|^4 \cdot |\keyAlphabet|^2)$.
\end{lemma}

\begin{proof} 
Constructing the vertices of $\keyGraph$ is in $\complexity(|\internalFunction|^2 + |\internalFunction| \cdot |\states|^2)$. Constructing its edges is in $\complexity(|\states|^4 \cdot |\keyAlphabet|^2)$ by Lemma~\ref{lem:sizeKeyGraph}.
\qed\end{proof}

The complexity announced in \Cref{prop:KeyGraph} follows from the previous \Cref{lem:ComputingReach,lem:ComputingBinState,lem:ComputingKeyGraph}.

\section{Generating a Counterexample from the Key Graph}\label{sec:counterexample_from_key_graph}
%=================================================================================================

\begin{lemma}
Let $\hypothesis$ be an automaton constructed by the learner. If the key graph $\keyGraph[\hypothesis]$ contains a path $((p_1,k_1,p'_1)\ldots (p_n,k_n,p'_n))$ with $p_1 = q_0$ such that $k_i = k_j$ for some $i \neq j$, then one can construct a word accepted by $\hypothesis$ that is not a valid JSON document.
\end{lemma}

\begin{proof}
We proceed exactly as in the proof of \Cref{lem:abstractKeyGraph}. In the key graph $\keyGraph[\hypothesis]$ of $\hypothesis$, let 
\begin{eqnarray} \label{eq:path:counterexample}
  ((p_1,k_1,p'_1)(p_2,k_2,p'_2)\ldots (p_n,k_n,p'_n))  
\end{eqnarray}
be a path with $p_1 = q_0$ such that $k_i = k_j$ for some $i \neq j$. Then, in the proof of \Cref{lem:abstractKeyGraph}, we proved that there exists in $\transitionSystem$ a path 
\begin{eqnarray} \label{eq:factor}
    \left\langle q_0, \emptyword \right\rangle \xrightarrow{t''} \left\langle q_0, \stack \right\rangle \xrightarrow{u} \left\langle p'_n, \stack \right\rangle \xrightarrow{t'} \left\langle q, \emptyword \right\rangle
\end{eqnarray}
with $q \in \finalStates$ such that $u = k_1v_1 \comma \ldots \comma k_n v_n$ is part of an object. This shows that the word $t''ut'$ is accepted by $\hypothesis$ and is not a valid document by the presence of the repeated keys $k_i, k_j$. Let us explain how to \emph{construct} this word $t''ut'$. 

For this purpose, we are going to use the witnesses introduced in Appendix~\ref{app:ComplexityKeyGraph}. First, thanks to the map $\wit[\hypothesis]$ associated with $\accRelation[\hypothesis]$ and computed in Appendix~\ref{app:ComplexityKeyGraph}, we can similarly enrich the vertices of $\keyGraph[\hypothesis]$ with witnesses: we assign a key-value pair $kv = \wit[\hypothesis](p,k,p')$ to each vertex $(p,k,p')$ of $\keyGraph[\hypothesis]$ such that $\left\langle q, \emptyword \right\rangle \xrightarrow{kv} \left\langle q', \emptyword \right\rangle$. Therefore, from a path in $\keyGraph[\hypothesis]$ like~\eqref{eq:path:counterexample}, we derive the witness $u = k_1v_1 \comma \cdots \comma k_n v_n$ such that $\left\langle q_0, \emptyword \right\rangle \xrightarrow{u} \left\langle p'_n, \emptyword \right\rangle$. It remains to extend this witness $u$ into a witness $t''ut'$ of a path like in~\eqref{eq:factor}. This is possible by noticing that $p'_n \in \setR$. Therefore with the computed witness $\wit[\hypothesis]'(q_0,p'_n) = (w,w')$ and the path $\left\langle q_0, \emptyword \right\rangle \xrightarrow{u} \left\langle p'_n, \emptyword \right\rangle$, we get the path $\left\langle q_0, \emptyword \right\rangle \xrightarrow{w} \left\langle q_0, \stackWord \right\rangle \xrightarrow{u} \left\langle p'_n, \stackWord \right\rangle \xrightarrow{w'} \left\langle q, \emptyword \right\rangle$ for some $q \in \finalStates$. The witness $wuw'$ of this path is the required counterexample.
\qed\end{proof}

\section{JSON Schemas Used for Experiments}\label{sec:schemas_implementation}

\Cref{fig:schemas_implementation:recursive,fig:schemas_implementation:all_types} give the first two evaluated JSON schemas of \Cref{sec:implementation}.

\begin{figure}
\begin{lstlisting}[language=JSON]
{
    "type": "object",
    "properties": {
        "name": { "type": "string" },
        "children": {
            "type": "array",
            "maxItems": 1,
            "items": { "$ref": "#" }
        }
    },
    "required": ["name"],
    "additionalProperties": false
}
\end{lstlisting}
    \captionof{lstfloat}{The recursive JSON schema.}%
    \label{fig:schemas_implementation:recursive}

\begin{lstlisting}[language=JSON]
{
    "Title": "Fast",
    "type": "object",
    "required": ["string", "double", "integer", "boolean", "object", "array"],
    "additionalProperties": false,
    "properties": {
        "string": { "type": "string" },
        "double": { "type": "number" },
        "integer": { "type": "integer" },
        "boolean": { "type": "boolean" },
        "object": {
            "type": "object",
            "required": ["anything"],
            "additionalProperties": false,
            "properties": {
                "anything": {
                    "type": ["number", "integer", "boolean", "string"]
                }
            }
        },
        "array": {
            "type": "array",
            "items": { "type": "string" },
            "minItems": 2
        }
    }
}
\end{lstlisting}
    \captionof{lstfloat}{The JSON schema accepting documents with all types.}%
    \label{fig:schemas_implementation:all_types}
\end{figure}

\section{Validation Results for the First Three Schemas}\label{app:more_results}
Results for the comparison of the classical and our new validation algorithms on the first three schemas of \Cref{sec:implementation:results} are given in \Crefrange{fig:benchmarks:validating:basicTypes}{fig:benchmarks:validating:vscode}.
Recall that blue crosses give the values for our algorithm, while the red circles stand for the classical algorithm.

\begin{figure}
    \centering
    \begin{subfigure}{0.45\textwidth}
        \includegraphics{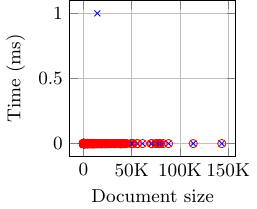}
        \caption{Time.}
    \end{subfigure}
    \hfill
    \begin{subfigure}{0.45\textwidth}
        \includegraphics{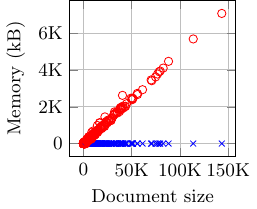}
        \caption{Memory.}
    \end{subfigure}
    \caption{Results of validation benchmarks for the recursive list schema.%
    %Blue crosses give the values for our algorithm, and red circles the values for the classical validator. Values are averaged out of ten experiments.
    }%
    \label{fig:benchmarks:validating:recursive}

    \begin{subfigure}{0.45\textwidth}
        \includegraphics{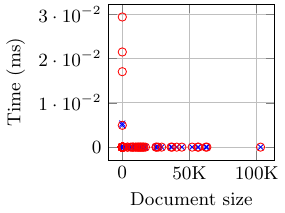}
        \caption{Time.}
    \end{subfigure}
    \hfill
    \begin{subfigure}{0.45\textwidth}
        \includegraphics{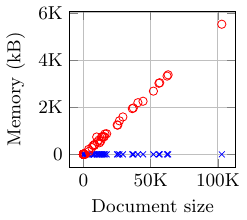}
        \caption{Memory.}
    \end{subfigure}
    \caption{Results of validation benchmarks for the schema iterating over the types of values.%
    %Blue crosses give the values for our algorithm, and red circles the values for the classical validator. Values are averaged out of ten experiments.
    }%
    \label{fig:benchmarks:validating:basicTypes}

    \begin{subfigure}{0.45\textwidth}
        \includegraphics{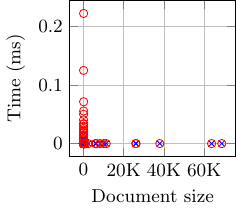}
        \caption{Time.}
    \end{subfigure}
    \hfill
    \begin{subfigure}{0.45\textwidth}
        \includegraphics{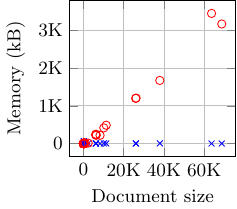}
        \caption{Memory.}
    \end{subfigure}
    \caption{Results of validation benchmarks for the snippet configuration of \emph{Visual Studio Code}.%
    %Blue crosses give the values for our algorithm, and red circles the values for the classical validator. Values are averaged out of ten experiments.
    }%
    \label{fig:benchmarks:validating:vscode}
\end{figure}

\end{document}